\documentclass[11pt]{amsart}
\usepackage{fullpage}
\usepackage{parskip}
\usepackage{array}
\usepackage{xcolor}
\usepackage{enumitem} % put this in your preamble
\usepackage{dcolumn}

\usepackage{amsmath, amssymb, amsthm}
\usepackage[foot]{amsaddr}

\newtheorem{thm}{Theorem}[section]
\newtheorem{lem}[thm]{Lemma}

\newtheorem{rem}[thm]{Remark}

\numberwithin{thm}{section}
\numberwithin{equation}{section}
\numberwithin{figure}{section}

\usepackage[colorlinks=true, linkcolor=red, citecolor=blue, urlcolor=cyan]{hyperref}

\usepackage{multirow, multicol, algorithm, algorithmic, tcolorbox}
\usepackage[mathic]{mathtools}
\tcbuselibrary{skins, theorems, breakable}

\usepackage{graphicx}
\usepackage{tikz}
\usepackage{pgfplots}
\usetikzlibrary{arrows.meta}
\usepgfplotslibrary{groupplots}
\pgfplotsset{compat=newest}

\usepackage[font=small, labelfont=bf]{caption}

\newcommand{\ie}{{\it i.e.}}

\providecommand{\sech}{}
\renewcommand{\sech}{\operatorname{sech}}

\newcommand{\R}{\mathbb R}
\newcommand{\C}{\mathbb C}
\newcommand{\sgn}{\operatorname{sgn}}

\usepackage[
  backend=biber,
  date=year,
  giveninits=true,
  sorting=nyt,
  natbib=true,
  maxcitenames=2,
  maxbibnames=10,
  url=false,
  doi=true,
  backref=false
]{biblatex}
\renewbibmacro{in:}{\ifentrytype{article}{}{\printtext{\bibstring{in}\intitlepunct}}}

\definecolor{limegreen}{RGB}{50,205,50}

\DeclareRobustCommand{\greenCirc}{%
  \tikz[baseline=-0.6ex]{%
    \filldraw[limegreen, draw=black, thin] (0,0) circle (0.6ex);}}

\DeclareRobustCommand{\blueX}{%
  \tikz[baseline=-0.6ex]{%
    \filldraw[white, draw=black, thin] (0,0) circle (0.6ex);
    \draw[blue, line width=0.8pt]
      (-0.5ex,-0.5ex) -- (0.5ex,0.5ex)
      (-0.5ex,0.5ex) -- (0.5ex,-0.5ex);}}
\DeclareRobustCommand{\blueXn}{%
  \tikz[baseline=-0.6ex]{%
    \draw[blue, line width=0.8pt]
      (-0.5ex,-0.5ex) -- (0.5ex,0.5ex)
      (-0.5ex,0.5ex) -- (0.5ex,-0.5ex);}}

\title{Resonance-induced nonlinear bound states}

\author{
Jackson C. Turner\textsuperscript{1} \and
Michael I. Weinstein\textsuperscript{1,2}
}

\thanks{\textsuperscript{1}Department of Applied Physics and Applied Mathematics, Columbia University, New York, NY}
\thanks{\textsuperscript{2}Department of Mathematics, Columbia University, New York, NY}

\email{jt3287@columbia.edu}
\email{miw2013@columbia.edu}

\date{\today}

\begin{document}
\begin{abstract}
We study nonlinear bound states --- time-harmonic and spatially decaying ($L^2$) solutions --- of the nonlinear Schr\"odinger / Gross--Pitaevskii equations (NLS/GP)  with a compactly supported linear potential. Such solutions are known to bifurcate from the $L^2$ bound states of an underlying  Schr\"odinger operator $H_V=-\partial_x^2+V$.  
In this article we prove an extension of this result: for the 1D NLS/GP,  nonlinear bound states also arise via bifurcation from the scattering resonance states  and transmission resonance states of $H_V$, associated with the poles and zeros, respectively, of the reflection coefficients, $r_\pm(k)$, of $H_V$.  The corresponding resonance states are non-decaying and only $L^2_{\rm loc}$. In contrast to nonlinear states arising from $L^2$ bound states of $H_V$, these resonance bifurcations initiate at a strictly positive $L^2$ threshold which is determined by the position of the complex scattering resonance pole or transmission resonance zero. 
\end{abstract}
\maketitle

\section{Introduction}

In this article we study bifurcations of nonlinear bound states of the focusing 1D nonlinear Schr\"odinger / Gross--Pitaevskii  equation (NLS/GP)
\begin{equation}\label{eq:time-dep-nlse}
    i \partial_t \Psi  = - \partial^2_x \Psi + V(x) \Psi -  |\Psi|^2 \Psi,
\end{equation} 
governing the evolution of the complex valued function:  $\Psi:(x,t)\in \mathbb{R}_x\times\mathbb{R}_t\mapsto \Psi(x,t)$. The linear potential $V(x)$ is real-valued and is  assumed to be compactly supported. We refer to  $H_V=-\partial_x^2+V(x)$  as the underlying linear Schr\"odinger operator.  PDEs of the type \eqref{eq:time-dep-nlse} are of central importance in the modeling of phenomena in nonlinear optics and plasma physics, where $\Psi$ plays of the role of a slowly varying envelope of a highly oscillatory electric field \cite{sulemsulem1999,fibich2015nonlinear}, and in many-body quantum systems, where $\Psi$ is used to construct the quantum many-body wave function in the mean field limit \cite{esy2010}.

It is common to consider more general focusing nonlinearities
of the form $-f(|\Psi|^2)$, such as the homogeneous power-law
family $-|\Psi|^{2\sigma}\Psi$,  with $\sigma>0$. The choice $\sigma_c=2$  is the $L^2(\mathbb R)$-critical  case \cite{sulemsulem1999,cazenave2003semilinear,fibich2015nonlinear}. In this work we treat the \emph{subcritical case} $\sigma = 1$ in detail, and indicate how the arguments extend to general~$\sigma$ and~$f$.

Nonlinear bound states are time-harmonic and spatially localized solutions of \eqref{eq:time-dep-nlse}. These play a role in the localization and transport of energy, and it is hence of great interest to understand the nonlinear bound states which arise for a given linear potential, $V$.

For initial data $ \Psi(x,0)=\Psi_0(x)\in H^1(\mathbb R)$,
 \eqref{eq:time-dep-nlse}  has a unique solution $\Psi(x,t)\in C(\mathbb R_t;H_x^1(\mathbb R))$; see for example \cite{sulemsulem1999,cazenave2003semilinear,oh1989cauchy}. Further, the following two functionals are time-invariant on solutions of NLS / GP: \begin{equation}\begin{aligned}\label{eq:conservation}
&\mathcal{H}[\Psi](t) \equiv \int_{\mathbb R} \Big(\ |\partial_x \Psi(x,t)|^2\ +\ V(x)|\Psi(x,t)|^2\ -\ \frac 12 |\Psi(x,t)|^4 \, \ \Big)\ dx\ =\ \mathcal{H}[\Psi](0)  ,\\
&\mathcal N[\Psi](t) \equiv \int_{\mathbb R}  | \Psi(x,t) |^2 \, dx\ =\ \mathcal N[\Psi](0) .
\end{aligned}\end{equation}
$\mathcal{H}$ is the Hamiltonian for NLS / GP, arising since \eqref{eq:time-dep-nlse} has a time-translation invariant Lagrangian, \ie\ equation
 \eqref{eq:time-dep-nlse} may be expressed as  $i\partial_t\Psi = \delta\mathcal{H}[\Psi,\overline\Psi]/\delta\overline\Psi$. The time-invariance of  $\mathcal{N}$  arises from the Lagrangian being invariant with respect to $\Psi\mapsto e^{i\theta}\Psi$.

\subsection{Nonlinear bound states}
Nonlinear bound states or nonlinear standing waves of NLS / GP are time-harmonic solutions, $\Psi(x,t)=e^{-iEt}\psi(x)$ of \eqref{eq:time-dep-nlse}, and thus $\psi$ satisfies the nonlinear problem boundary value problem on $\mathbb R$,
\begin{align}
    \Big(-\partial_x^2 + V\Big)\psi -|\psi|^2\psi &= E\psi,\quad    \psi \in H^1(\mathbb R).
\label{eq:V-nls-full}\end{align}
We focus on real-valued solutions of \eqref{eq:V-nls-full}.

For the special case of $V=0$, \eqref{eq:V-nls-full} reduces to
\begin{align}
-\partial_x^2\psi -|\psi|^2\psi &= E\psi,\quad    \psi \in H^1(\mathbb R),
\label{eq:0-nls-full}\end{align}
the equation for the ``1-soliton'' profile, described as follows.
\begin{thm}[The 1-soliton of cubic NLS]\label{thm:1soliton}
 Fix any $E<0$.
\begin{enumerate}
  \item The nonlinear eigenvalue problem \eqref{eq:0-nls-full} has a unique positive and even solution,  $\mathcal{S}(x,E)$,  which is monotonically decreasing away from $x=0$:
  \begin{equation}\label{def:sech-thing}
    \mathcal{S}(x; E) \coloneqq \sqrt{-2 E} \, \sech \! \left(\sqrt{-E} \, x\right) .
\end{equation}
\item Any real-valued and decaying solution of \eqref{eq:0-nls-full} is of the form
 $\psi(x) = +\mathcal{S}(x-x_0,E)$ or 
  $\psi(x) = -\mathcal{S}(x-x_0,E)$
  for some $x_0\in\mathbb R$.
\item  $ \mathcal{N}\big[\mathcal{S}(\cdot,E)\big] =  \int_{\mathbb R}\mathcal{S}^2(z,E)\ dz\ =  4 \sqrt{-E} $.
    \end{enumerate}
\end{thm}

\subsection{Nonlinear bound states via bifurcation} Note, by Theorem \ref{thm:1soliton}, as $E\uparrow 0$, $\mathcal{N}\big[\mathcal{S}(\cdot,E)\big] = 4 \sqrt{-E}$ decreases to zero. So one can schematically represent  the solutions of \eqref{eq:0-nls-full} in a plot of $\mathcal{N}$ vs.\ $E$, for $E<0$, as bifurcating from the zero solution, $\mathcal{N}[\psi]=0$, from energy $E=0$, at the endpoint (threshold) of the continuous spectrum of $-\partial_x^2$. 

To motivate the perspective of  this article, consider the linear Schr\"odinger operator, $H^\varepsilon=-\partial_x^2 + \varepsilon V(x)$.
If $V$ is a potential well ($V\le0$ and non-trivial), then for $\varepsilon$ is positive (no matter how small),
$H_V$ has a $L^2$ eigenstate with strictly negative energy in its point spectrum  \cite{simon1976bound}.
 The limiting (unperturbed) operator $H_0=-\partial_x^2$ has a threshold resonance at the bottom of its continuous spectrum, at $E=0$, and a corresponding uniformly bounded threshold resonance mode: $\psi(x)\equiv1$. For $\varepsilon>0$ and small, this threshold mode deforms into a localized ($L^2$) (slowly varying) mode of  $H^\varepsilon$ with energy of order $\varepsilon^2$ :
\[ \psi^\varepsilon(x)\sim 1\times \exp\Big(-c\ \varepsilon \ \lvert x \rvert \Big),\quad E^\varepsilon \sim -c^2\ \varepsilon^2. \]

Heuristically speaking, a soliton nonlinear bound  state is  self-induced by a shallow (self-consistent) nonlinear potential $-|\psi|^2$; the attractive nonlinear potential ``pulls'' the threshold resonance energy to strictly negative values.
In fact, numerical schemes for computing solitons are based on computing the linear ground state and  iteratively updating the self-consistent potential
\cite{pelinovsky2004convergence,petviashvili1976equation}.

The underlying linear operator of NLS/GP \eqref{eq:V-nls-full} is $H_V=-\partial_x^2 + V$. In \cite{rose1988bound} it was proved that nonlinear bound states of NLS/GP bifurcate from the discrete eigenvalues  (point spectrum) of $H_V$, and further that the family of nonlinear ground states are orbitally stable.  

The preceding is to indicate how nonlinear bound states may arise from scattering / spectral features of an underlying linear operator --- from a threshold resonance of $H_0=-\partial_x^2$ or from the point spectrum of $H_V$.  The article \cite{rose1988bound} speculated on the possibility  that spectral features of $H_V$  beyond $L^2$ eigenstates, e.g.\ resonances, may play a role in the nucleation of nonlinear bound states. This is the point of departure for the present article.

{\it We present rigorous results and numerical simulations on the bifurcation of   nonlinear bound states from:  bound state poles, 
scattering resonance poles, threshold resonances, and transmission resonances of the underlying linear operator $H_V$. Further, 
the branches of nonlinear bound states which are seeded by scattering resonance poles and transmission resonances exhibit a strictly positive $L^2$ (nonlinear) threshold; see schematic Figures \ref{fig:profiles-norm}, as well as Figures \ref{fig:boundstate}, \ref{fig:resonance}, and \ref{fig:transmission}.}

\subsection{An illustrative example}\label{sec:dirac_example}

Consider the  Dirac delta potential: $ V(x) = \alpha \delta(x)$, where $\alpha \in \mathbb{R}$.
The case $\alpha>0$ corresponds to a potential barrier and the case $\alpha<0$ a potential well.
  We construct nonlinear bound states of \eqref{eq:V-nls-full}. Away from $x=0$, the potential vanishes and solutions take the form of translated solitons:
\[
    \psi_E(x) = 
    \begin{cases} 
        \mathcal{S}(x - x_{\rm R}; E), & x > 0, \\
        \mathcal{S}(x - x_{\rm L}; E), & x < 0,
    \end{cases}
\]
where $E<0$. The centerings $x_{\rm R}$ and $x_{\rm L}$ are determined by a continuity condition on $\psi(x)$ and a jump condition  on $\psi^\prime$ at $x=0$:
\[
    x_{\rm R} = -x_{\rm L} = \frac{1}{\sqrt{|E|}} 
      \tanh^{-1}\!\left( \frac{\alpha}{2 \sqrt{|E|}} \right).
\]
The requirement that $x_{\rm L}$ and $x_{\rm R}$ take on  real values imposes the additional condition constraint on $E$: $|\alpha|/2\sqrt{|E|}<1$. Therefore $E< E_\star(\alpha) = -\alpha^2/4$. If $\alpha<0$ (potential well),  then as $E\uparrow E_\star$ we have $x_{\rm R}\to -\infty$ and $x_{\rm L}\to \infty$ (Figure  \ref{fig:profiles-norm} bottom plot).
And if $\alpha>0$ (potential barrier),  then as $E\uparrow E_\star(\alpha)$ we have $x_{\rm R}\to +\infty$ and $x_{\rm L}\to - \infty$ (Figure  \ref{fig:profiles-norm} top plot).

\begin{figure}[!ht]
    \centering
    \includegraphics[width=.45\linewidth, page = 1]{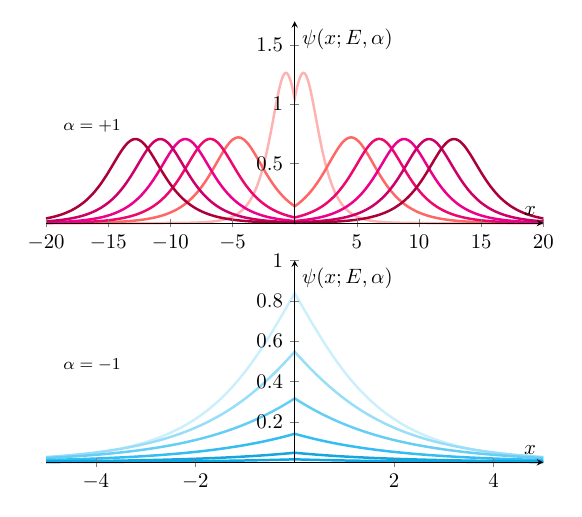}
    \includegraphics[width=.35\linewidth, page = 2]{Figures/delta_example.pdf}
\caption{
Profiles $\psi(x;E,\alpha)$ for $V=\alpha\delta(x)$ (left) and norms $\mathcal N[\psi_E]$ (right). 
Top: potential barrier ($\alpha>0$), with solutions shaded from light to deep red as $E\to E_\star$; the profiles approach the threshold logarithmically, producing constant spatial shifts. 
Bottom: potential well ($\alpha<0$), shaded from light to deep cyan. 
Right: the $(E,\mathcal N)$ branches terminate at $E_\star$, with an excitation threshold only for $\alpha>0$.
}
    
    \label{fig:profiles-norm}
\end{figure}

For $E\downarrow E_\star(\alpha)$,  we observe a strictly positive threshold $L^2$ norm in the case 
of a repulsive delta potential ($\alpha>0$)
and no such threshold for $\alpha<0$:
\begin{equation}
    \mathcal{N}[\psi_E] \;\to\;
    \begin{cases}
        8\sqrt{-E_\star(\alpha)}, & \alpha > 0, \\
        0, & \alpha < 0\, .
    \end{cases}
    \label{eq:ex-thresh}\end{equation}
As we shall see, this dichotomy can be understood as a bifurcation from a scattering resonance of $H_V$ when $\alpha>0$,  versus bifurcation from a bound state  of $H_V$ when $\alpha<0$. The analytical framework of bound states and resonances is reviewed in Section~\ref{sec:scattering} to set up for our main results, which are stated in Section~\ref{sec:bifurcations}.

\subsection{Outline of the article}

Section~\ref{sec:scattering} reviews basic scattering theory for $H_V$, in particular the notions of  bound state pole, scattering resonance poles, transmission resonances, and threshold resonances.  In Section~\ref{sec:NOGO} we prove a threshold-type result: when $H_V$ has neither a bound state nor a zero-energy resonance, nonlinear bound states are confined away from the axes in the $(E,\mathcal N)$ plane. In Section~\ref{sec:bifurcations} we state and prove the bifurcation of nonlinear bound states from different features of the linear scattering problem for $H_V$:
non-zero bound state poles in the upper half plane,
non-zero scattering resonance poles on the imaginary axis (also called anti-bound states), and non-zero transmission resonances on the imaginary axis. In Section~\ref{sec:bifurcation-character} we present schematics and numerical simulations detailing these different scenarios.
In Section \ref{sec:future} we provide a summary and discussion of future directions.

\subsection*{Acknowledgements} 
MIW acknowledges the very stimulating collaboration with Harvey Rose \cite{rose1988bound} during the 1980s, where the questions explored in the current work originate.  The authors thank Henri Berestycki,  Panayotis Kevrekidis, Eduard Kirr, Mikael Rechtsman and the research group of Sebastian Will for stimulating discussions. 
This work was supported in part by NSF grants: DMS-1908657, DMS-1937254 and DMS-2510769, and Simons Foundation Math+X Investigator Award \#376319. Part of this research was carried out during the 2023-24 academic year, when MIW was a Visiting Member in the School of Mathematics, Institute of Advanced Study, Princeton, supported by the Charles Simonyi Endowment, and a Visiting Fellow in the Department of Mathematics at Princeton University.

\section{Some scattering theory}\label{sec:scattering}

Here, we outline some basic scattering theory; see, for example, \cite{deift1979inverse,korotyaev2005inverse,reed1979iii,dyatlov2019mathematical}. We make the following assumptions on $V$:
\begin{enumerate}
    \item $V\in L^1(\mathbb R)$ and real-valued.
    \item ${\rm supp}(V) = [-b,b]$, where $b>0$.
\end{enumerate}

Outside the support of  $V$, the solutions of $(H_V-k^2)\psi=0$ are linear combinations of the exponentials $e^{ikx}$ and $e^{-ikx}$.
For $k\in\mathbb C$, the {\it Jost solutions}, $f_\pm(x,k)$, are defined to be the unique solutions 
of:
    \begin{align*}
        H_V f_+(x,k) &= k^2 f_+(x,k),\quad f_+(x,k)=e^{ikx},\quad \textrm{for $x>b$}\\
        H_V f_-(x,k) &= k^2 f_-(x,k),\quad f_-(x,k)=e^{-ikx},\quad \textrm{for $x<-b$}.
    \end{align*}

For $k\in\mathbb{C}\setminus\{0\}$, $f_-(x,k)$ and $f_-(x,-k)$ are linearly independent and hence there exist $b_-(k)$ and $a_-(k)$ such that:
\begin{equation} f_+(x,k) = b_-(k) f_-(x,k) + a_-(k) f_-(x,-k).  \label{eq:jost-f_+}
\end{equation}
And similarly, for $k\in\mathbb{C}\setminus\{0\}$
\begin{equation} f_-(x,k) = b_+(k) f_+(x,k) + a_+(k) f_+(x,-k).  \label{eq:jost-f_-}
\end{equation}
Relations for $f'_\pm(x,k)$ in terms of $f'_\mp(x,k)$ and $f'_\mp(x,-k)$ follow by differentiation. To solve for $b_\pm(k)$ and $a_\pm(k)$, first introduce the two ($x-$independent) Wronskian determinants: 
\begin{align}
    w(k) &\equiv f_-(x,k)\,f_+'(x,k)\;-\;f_-'(x,k)\,f_+(x,k)\quad  \textrm{and}\\   s_\pm(k) &\equiv
    \;f_+(x,\mp k)\,f_-'(x,\pm k)\;-\;f_+'(x,\mp k)\,f_-(x,\pm k).
\end{align}
We have:
 \begin{align}
     b_-(k) = \frac{s_-(k)}{2ik},\quad a_-(k) = \frac{w(k)}{2ik}\\
     b_+(k) = \frac{s_+(k)}{2ik},\quad a_+(k) = \frac{w(k)}{2ik}
 \end{align}
 \begin{equation} f_+(x,k) = \frac{s_-(k)}{2ik} f_-(x,k) + \frac{w(k)}{2ik} f_-(x,-k).  \label{eq:jost-f_+ab}
\end{equation}
And similarly, for $k\in\mathbb{C}\setminus\{0\}$
\begin{equation} f_-(x,k) = \frac{s_+(k)}{2ik} f_+(x,k) + \frac{w(k)}{2ik} f_+(x,-k).  \label{eq:jost-f_-ab}
\end{equation}

Since $V$ has compact support, for fixed $x$, the maps $k\mapsto f_\pm(x,k)$,  $s_\pm(k)$ and $r(k)$ are entire functions of $k\in\mathbb{C}$. 
By \eqref{eq:jost-f_+ab} and \eqref{eq:jost-f_-ab}, we have that $s_-(k) = 0$ iff $s_+(-k) = 0$. Furthermore,  $ s_-(k) = 0$ iff $f_+(x,k) = \frac{w(k)}{2ik} f_-(x,-k)$ and  $s_+(k) = 0$ iff  $f_-(x,k) =  \frac{w(k)}{2ik} f_+(x,-k)$. Hence,
 \[ \textrm{$s_-(k)$ and $w(k)$ cannot vanish together at any point in $\mathbb C\setminus\{0\}$},\]
 and similarly for $s_+(k)$ and $w(k)$.
 Below, we shall be interested in the zeros of $s_-(k)$ and $s_+(k)$.
Since $s_+(k)=s_-(-k)$, it suffices to restrict attention to $s_-(k)$.

Next, we introduce the transmission and reflection coefficients 
 \[ t(k) \equiv \frac{2ik}{w(k)} \quad {\rm and}\quad  \ r_\pm(k) \equiv \frac{s_\pm(k)}{w(k)} .
 \]
 From \eqref{eq:jost-f_+} we have that
 \begin{equation}
   f_-(x,-k) +   r_-(k) f_-(x,k)  =  t(k) f_+(x,k),\quad k\in\mathbb C\setminus\{0\}.
 \end{equation}
For $k\in\mathbb R\setminus\{0\}$, this relation encodes the scattering process of an incident plane wave $e^{ikx}$ in the region $x<-b$ inducing a refected wave
  $r_-(k)e^{-ikx}$ in this region and a transmitted wave
   $t(k) e^{ikx}$ in the region $x>a$.
From  \eqref{eq:jost-f_-}, we similarly have that 
 \begin{equation}
   f_+(x,-k) +   r_+(k) f_+(x,k)  =  t(k) f_-(x,k)
   ,\quad k\in\mathbb C\setminus\{0\}, 
 \end{equation}
 and the corresponding interpretation,  for $k\in\mathbb R\setminus\{0\}$, of the scattering of an incident wave $e^{ikx}$ from the region $x>a$. Finally, a consequence of the equation $H_V\psi=k^2\psi$ is the relation (conservation of probability):
 \begin{equation} \lvert r_\pm(k) \rvert ^2  + \lvert t(k) \rvert ^2 =1\quad \textrm{ for $k\in\mathbb{R}$.} \label{eq:r2t2}\end{equation}

\subsubsection*{Bound state poles and scattering resonance poles} Let $\mathbb C_+$ and $\mathbb C_-$ denote, respectively, the \underline{open} upper and lower half planes in $\mathbb C$. Consider now the zeros of $k\in\mathbb C_\pm\mapsto w(k)$.

\begin{enumerate}
    \item Let  $k_\star\in \mathbb C_+$ and $w(k_\star)=0$. By \eqref{eq:jost-f_+ab}, $f_+(x,k_\star) = \frac{s_-(k_\star)}{2ik_\star} f_-(x,k_\star)$ and is exponentially decaying as $x\to\pm\infty$. Hence, $f_+(x,k_\star)\in L^2(\mathbb R)$, and is therefore an eigenfunction of $H_V$ with  eigenvalue $k_\star^2$. By self-adjointness of $H_V$, $k_\star^2$ is real. Hence, $k_\star=i\kappa_\star$, with $\kappa_\star>0$ and the corrresponding eigenvalue of $H_V$ is $E=k_\star^2=-\kappa_\star^2<0$.
    \item If $k_\star\in \mathbb C_-$ and $w(k_\star)=0$, then $f_+(x,k_\star)$ solves $Hf_+(x,k_\star)=\kappa_\star^2f_+(x,k_\star)$, but $f_+(x,k_\star)$ is exponentially growing as $x\to + \infty$ and as $x\to - \infty$. These states are called {\it scattering resonance modes} or {\it quasi-normal modes}. 
\end{enumerate}

 Hence, using \eqref{eq:r2t2}, $k_\star\in{\mathbb C}_+\cup{\mathbb C}_- \cup{\{0\}} $ and $w(k_\star)=0$ if and only if there is a non-trivial solution, $\big(k_\star,\psi\big)$, of the eigenvalue problem for $H_V$ with {\it outgoing  radiation conditions} imposed at $x=\pm b$:
\begin{subequations}
\label{eq:outgoing-rad-prob}
\begin{align}
 H_V \psi &= k^2_\star \psi,\quad x\in\mathbb R\\
 \partial_x \psi &= ik_\star \psi,\quad x=b\\
  \partial_x \psi &= -ik_\star \psi,\quad x=-b.
 \end{align}
\end{subequations}
If $(k_\star,\psi)$ is a scattering resonance pair, then $k_\star\in\mathbb C_-$ may be on or off the imaginary axis. 
 Let $k_\star=i\kappa_\star$, with $\kappa_\star<0$, denote a scattering resonance pole, which falls on the negative imaginary axis. The associated exponentially growing (and non-oscillatory) state $f_+(x,k_\star)$ is called an {\it anti-bound state} or {\it virtual state} \cite{BreitWigner1936,Islam1966,LandauLifshitzQM,Newton1982,OhanianGinsburg1974}.

 Finally, we remark that the resolvent $(H_V-k^2)^{-1}$, as an operator on $L^2(\mathbb R)$, can be represented as an integral operator with Green's kernel $G_V(x,y,k)=g(x,y,k)/w(k)$:
\begin{equation} \label{eq:green1}
\bigl(H_V - k^2\bigr)^{-1}f(x)
= \int_{\mathbb R} G_V(x,y;k)\,f(y)\,dy,
\qquad
G_V(x,y;k) = \frac{1}{w(k)}
\begin{cases}
f_+(x,k)\,f_-(y,k), & x \ge y, \\[6pt]
f_+(y,k)\,f_-(x,k), & x < y.
\end{cases}
\end{equation}
  Here, for each fixed $x, y \in\mathbb{R}$, the mapping $k\mapsto g(x,y,k)$ is an entire function, and hence
   the Green's kernel, $k\mapsto G(x,y,k)$ 
 is meromorphic in the complex plane, with poles occurring precisely at the zeros of $w(k)$. 

 Since poles which occur in the upper half plane correspond to bound states of $H_V$, we call zeros of $w(k)$ in $\mathbb C_+$ {\it bound state poles}. Poles of the resolvent kernel which occur in the lower half plane, correspond to solutions of the outgoing radiation problem 
\eqref{eq:outgoing-rad-prob} and are called
 {\it scattering resonance poles}. 

 It may happen that $w(0)=0$. In this case, the boundary value problem reduces to $H_V\psi=0$
  with Neumann boundary conditions $\partial_x\psi(-b)=0=\partial_x\psi(b) $. In this case, we say that $E=k_\star^2=0^2=0$ is a {\it threshold resonance} and $\psi(x)$ is constant outside the support of $V$.

\subsubsection*{Transmission resonances} Now let $k_\star\in\mathbb{R}$ be such that  $s_-(k_\star)=0$. By \eqref{eq:jost-f_+ab}, 
 $f_+(x,k_\star) = \frac{w(k_\star)}{2ik_\star} f_-(x,-k_\star)$. By the conservation law \eqref{eq:r2t2} we have
  that $\lvert t(k_\star) \rvert = \big\lvert w(k_\star)/2ik_\star\big \rvert =1$. Hence, an incoming plane wave $e^{ik_\star x}$ of amplitude equal to $1$ interacts with the potential and transmits as a phase shift plane wave, $e^{i\theta}\ e^{ik_\star x} $, of amplitude equal to $1$. We therefore call $k_\star$ a {\it transmission resonance}.

 We have that  $k_\star\in\mathbb{C}$ and $s_-(k_\star)=0$ iff there is a non-trivial solution $(k_\star,\psi)$  of the (right) transmission problem with right transmission (non-reflecting) boundary conditions at $x=\pm b$:
 \begin{subequations}
\label{eq:right-trans-prob}
\begin{align}
 H_V \psi &= k^2_\star \psi,\quad x\in\mathbb R\\
 \partial_x \psi &= ik_\star \psi,\quad x=b\\
  \partial_x \psi &= ik_\star \psi,\quad x=-b.
 \end{align}
\end{subequations}

Similarly, $k_\star\in\mathbb{R}$ with  $s_+(k_\star)=0$ correspond to solutions of the left transmission problem,
left transmission (non-reflecting) boundary conditions at $x=-b$ and $x=b$.

We shall call the zeros of $s_-(k)$ {\it right transmission resonances} and those of $s_+(k)$ {\it left transmission resonances}. We shall apply this terminology to \underline{any} zero $k_\star\in\mathbb C$  of $s_-(k)$ or $s_+(k)$.

\section{\texorpdfstring{Constraints on $(E,\mathcal{N})$ if $H_V$ has no bound states or threshold resonances}
{Constraints on (E,N) if HV has no bound states or threshold resonances}}
\label{sec:NOGO}

Before turning to bifurcations in Section \ref{sec:bifurcations}, we prove a general constraint on nonlinear bound states $\big(E,\mathcal{N}[\psi_E]\big)$; see Figure~\ref{fig:nogo}.

\begin{thm}\label{thm:nogo}
Let $V\in L^1(\mathbb{R})$ have compact support, and let $H_V=-\partial_x^2+V(x)$
act in $L^2(\mathbb R)$.  
\begin{enumerate}
    \item If $H_V$ has no threshold resonance, then there exists a constant
    $E_{\mathrm{thr}}<0$, depending only on $V$, such that any nontrivial solution
    $(E,\psi)$ of \eqref{eq:V-nls-full} satisfies $E\le E_{\mathrm{thr}}$.
    \item If, in addition, $H_V$ has no bound state, then there exists
    $\mathcal N_{\mathrm{thr}}>0$, depending only on $V$, such that
    $\mathcal N[\psi]\ge \mathcal N_{\mathrm{thr}}$.
\end{enumerate}
\end{thm}

\begin{rem}[Generalization to power nonlinearities]
The conclusions of Theorem~\ref{thm:nogo} extend to focusing
nonlinearities $-|\psi|^{2\sigma}\psi$ with analogous proofs, in the follow ways:
\begin{enumerate}
\item If $H_V$ has no threshold resonance, then for all $\sigma > 0$,
$E_{\mathrm{thr}}(V,\sigma)<0$.
\item When $H_V$ has neither a bound state nor a threshold resonance and
$\sigma\le2$ (the subcritical or critical case),
one also has $\mathcal N_{\mathrm{thr}}(V,\sigma)>0$, whereas for $\sigma>2$ (the supercritical case), $\mathcal N_{\rm thr}$ can be 0.
\end{enumerate}
These scenarios are illustrated in Figure~\ref{fig:nogo}. 
\end{rem}

\begin{figure}[!ht]
    \centering
    \includegraphics[width=0.47\linewidth,page=1]{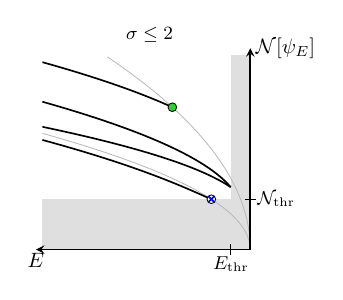}\hfill
    \includegraphics[width=0.47\linewidth,page=2]{Figures/nogo.pdf}
    \caption{
        Schematic illustration of Theorem~\ref{thm:nogo} for
        nonlinearities $-|\psi|^{2\sigma}\psi$. The shaded grey areas denote forbidden regions in the
        $\bigl(\mathcal N,E\bigr)$-plane.
        Left: subcritical and critical case $\sigma\le2$ when $H_V$ has no bound state or threshold resonance. Right: supercritical case $\sigma>2$ when $H_V$ has no threshold resonance.
        The symbols \protect{\greenCirc}\ and \protect{\blueX}\ mark branches bifurcating from
        \emph{scattering} and \emph{transmission} resonances of $H_V$
        on $i\mathbb{R}$; see 
        Theorem~\ref{thm:bifurcations}.
        Each bifurcating branch emerges from the faint gray ``free soliton guide curves''
        $E\mapsto \mathcal N[\mathcal S_{E,\sigma}]$ and
        $E\mapsto 2\,\mathcal N[\mathcal S_{E,\sigma}]$,
        where $\mathcal S_{E,\sigma}$ denotes the fundamental 1-soliton.
    }
    \label{fig:nogo}
\end{figure}

\begin{lem}\label{lem:bound-4}
Let $\{u_n\} \subset H^1(\mathbb R)$ be a sequence of solutions to
\begin{equation}\label{eq:uk}
    -u_n'' + V u_n - \eta_n \lvert u_n\rvert^2 u_n = E_n u_n
\end{equation}
for some sequence $E_n < 0$. Assume that
\[
    \lVert u_n\rVert_{L^2} = 1, 
    \qquad \eta_n \downarrow 0.
\]
Then $\{\|u_n\|_{L^4}\}$ is bounded.
\end{lem}

\begin{proof}[Proof of Lemma~\ref{lem:bound-4}]
Our proof is by contradiction. Suppose  $\lVert u_n\rVert_{L^4}\to\infty$.  
Then, since
\[
    \lVert u_n\rVert_{L^4}^4 
    \leq C\, \lVert u_n'\rVert\, \lVert u_n\rVert^3 
    = C\,\lVert u_n'\rVert,
\]
we have  $\lVert u_n'\rVert\to\infty$.  
Further, from \eqref{eq:uk} we have the energy identity:
\begin{equation}\label{eq:energy}
    \lVert u_n'\rVert^2 
    + \int V u_n^2 \,dx 
    - \eta_n^2 \lVert u_n\rVert_{L^4}^4 
    = E_n \lVert u_n\rVert^2 
    = E_n.
\end{equation}
We next obtain a lower bound for the second term in \eqref{eq:energy} as follows.  
Since $\lVert u_n\rVert =1$ and 
$\lVert u_n\rVert_\infty^2 \le 2\,\lVert u_n\rVert\,\lVert u_n'\rVert = 2\,\lVert u_n'\rVert$, we have 
\[
    \int V u_n^2 \,dx 
    \geq - \lVert V\rVert_{L^1}\,\lVert u_n\rVert_\infty^2 
    \geq - 2\,\lVert V\rVert_{L^1}\,\lVert u_n'\rVert.
\]
Hence,
\[
    E_n 
    \geq \lVert u_n'\rVert^2 
    - \bigl(2\lVert V\rVert_{L^1}+2\eta_n^2\bigr)\lVert u_n'\rVert.
\]
Since $\{\eta_n\}$ is bounded, and the right-hand side tends to infinity, we conclude $E_n\to\infty$, a contradiction.  
Hence,  $\{\lVert u_n\rVert_{L^4}\}$ is bounded.
\end{proof}

\begin{proof}[Proof of Theorem~\ref{thm:nogo}]
If \eqref{eq:V-nls-full} has no solutions for the given $V$, the claim is immediate.  
Otherwise, let $H_V$ have no threshold resonance and let $(\psi_n,E_n)$ be a sequence of solutions with $E_n \uparrow 0$.  
On the left tail $\psi_n(x)=\mathcal S(x-y_n;E_n)$, so $\psi_n(-b)\to 0$.  
Writing $\psi_n=\eta_n u_n$ with $\eta_n=\psi_n(-b)\to 0$, the compact support of $V$ ensures $u_n$ solves
\[
   -u_n'' + V u_n = E_n u_n + \eta_n^2 u_n^3 , \qquad u_n(\pm b)\neq 0.
\]
As $\eta_n\to 0$ and $E_n\to 0$, elliptic estimates and $V \in L^1$ give (up to subsequence) $u_n\to u$ in $H^1_{\mathrm{loc}}(\R)$ with 
\[
   -u'' + V u = 0, \qquad u'(\pm b)=0, \quad u(\pm b)\neq 0.
\]
Thus $u$ is a threshold resonance of $H_V$, a contradiction.
 
Hence there exists $E_{\rm thr}<0$, depending only on $V$, such that every solution satisfies $E\le E_{\rm thr}$.  

Next, let $H_V$ also have no bound state, and suppose for contradiction that there exists a sequence of solutions $(\psi_n,E_n)$ with 
$\mathcal N[\psi_n]\to 0$.  
Set $\phi_n=\psi_n/\lVert\psi_n\rVert$, so that $\lVert\phi_n\rVert=1$ and
\[
   -\phi_n'' + V\phi_n - \mathcal N[\psi_n]\phi_n^3 = E_n \phi_n.
\]
By Lemma~\ref{lem:bound-4}, the sequence $\{\lVert\phi_n\rVert_{L^4}\}$ is uniformly bounded.  
The weak form implies
\begin{equation}\label{eq:useNL}
   E_n = \lVert\phi_n'\rVert^2 + \int V\lvert\phi_n\rvert^2 \, dx - \mathcal N[\psi_n]\lVert\phi_n\rVert_{L^4}^4 < 0.
\end{equation}
If $\mathcal N[\psi_n]\to 0$, the last term vanishes, so
\begin{equation} \label{eq:supEn}
   \limsup_{n} E_n = \limsup_{n} \Big( \lVert\phi_n'\rVert^2 + \int V\lvert\phi_n\rvert^2 \, dx \Big).
\end{equation}
On one hand, by the definition of the threshold, we have
\[
   \limsup_{n} E_n \le E_{\rm thr} < 0.
\]
On the other hand, the right-hand side of \eqref{eq:supEn} is the linear Schr\"odinger energy of a sequence of normalized functions, so the limit can only be negative if $H_V$ has a bound state.  
Since $H_V$ has no eigenvalues by hypothesis, this is a contradiction.  

Therefore there exists $\mathcal N_{\rm thr}>0$, depending only on $V$, such that every solution satisfies $\mathcal N[\psi]\ge \mathcal N_{\rm thr}$.  This proves Theorem \ref{thm:nogo}.
\end{proof}

\section{\texorpdfstring{Nonlinear bound states of NLS/GP from bound states and resonances of $H_V$}{Nonlinear bound states of NLS/GP from bound states and resonances of HV}}
\label{sec:bifurcations}

We turn to the nonlinear eigenvalue problem \eqref{eq:0-nls-full}, whose solutions are nonlinear bound states of NLS / GP. 
Since we seek solutions which are seeded by ``linear scattering data'' of $H_V$, we set
\begin{equation}  \psi(x) = \eta\, u(x),\label{eq:psietau}
\end{equation}
where $\eta$ will later be taken to be sufficiently small.
Then, $u$ satisfies the nonlinear eigenvalue problem:
  \begin{subequations}
  \label{eq:inner-outer}
  \begin{align}
      \Big(-\partial_x^2 + V -\eta^2 u^2\Big)u &= Eu,\quad \lvert x\rvert < b, \\
       \Big(-\partial_x^2 -\eta^2 u^2\Big) u &= Eu,\quad \lvert x\rvert > b, 
  \end{align}
  \end{subequations}
  with continuity conditions at $x=\pm b$ on $u$ and $\partial_x u$:
  \begin{equation} \big[u\big]\Big|_{\lvert x\rvert=b} =0\quad{\rm and}\quad   \big[\partial_x u\big]\Big|_{\lvert x\rvert=b} =0,
  \label{eq:continuity}
  \end{equation}
  and decay as $x$ tends to infinity:
\begin{equation} u(x)\to0\quad {\rm as}\quad \lvert x\rvert\to\infty.\label{decay}
\end{equation}
Clearly, $E<0$ is necessary for a non-trivial solution to exist.
\medskip

\subsubsection*{Toward a solution of \eqref{eq:inner-outer}} 
In the region $\lvert x\rvert > b$, any solution $u(x)$ which satisfies the decay condition \eqref{decay} must be of the form 
\begin{align}
    u(x)= \begin{cases} 
    \pm\mathcal{S}(x-x_{\rm L};E), & x<-b,\\[3pt]
     \pm\mathcal{S}(x-x_{\rm R};E), & x>+b,
     \end{cases}
\end{align}
for some choices of signs and real parameters $x_{\rm L}$ and $x_{\rm R}$.

Further, from the autonomous ODE satisfied for $\lvert x\rvert >b$ we have the identity
\[
\mathfrak{E}[u](x)\equiv  \big(u'(x)\big)^2  + \frac{\eta^2}{2} u(x)^4 + E\,u(x)^2 = 0,\quad\textrm{for}\quad  \lvert x\rvert >b.
\]
Hence, by the continuity constraint at $x=\pm b$, a solution $u(x)$ of \eqref{eq:inner-outer}, restricted to $\lvert x\rvert \le b$, must satisfy a boundary value problem with a nonlinear boundary condition:
\begin{subequations}\label{eq:nonlinear-bvp_full}
    \begin{align}\label{eq:nonlinear-bvp}
\Big(-\partial_x^2 + V(x) - \varepsilon\, u(x)^2\Big)u(x) &= E\,u(x),\quad \lvert x\rvert < b,\\
   \mathfrak{E}[u](\pm b)\ &=\ 0. \label{eq:nonlinear-bvp_BC}
\end{align}
\end{subequations}
Here, we have set $\varepsilon=\eta^2$. 
The system \eqref{eq:nonlinear-bvp} is equivalent to \eqref{eq:inner-outer} since any solution of \eqref{eq:nonlinear-bvp} can be extended to a function on all $\mathbb R$ which satisfies \eqref{eq:inner-outer}; see part (2) of Theorem \ref{thm:bifurcations} below.

Let us formally seek a solution of \eqref{eq:nonlinear-bvp} as a formal power series in $\varepsilon$:
\begin{subequations}\label{eq:power_series}
\begin{align}
    u(x) &= u_0(x) + \varepsilon u_1(x) + \dots \label{eq:ps-u}\\
    E &= E_0 + \varepsilon E_1 + \dots.\label{eq:ps-E}
\end{align}
\end{subequations}
Subsitution into \eqref{eq:nonlinear-bvp} we find
\begin{subequations}
\begin{align} \Big(-\partial_x^2+V\Big)u_0 &= E_0 u_0 \label{eq:u0-eqn}\\
\Big[\ \big(u_0'\big)^2 + E_0\ u_0^2\ \Big]\Big|_{x=\pm b} &=0.
\label{eq:u0-BC}\end{align}
\end{subequations}
We set $E_0=k_0^2$ and re-express the boundary condition \eqref{eq:u0-BC} as 
\begin{equation}
\Big( u_0' + ik_0 u_0\Big)\ \Big( u_0'  - ik_0 u_0\Big)\Big|_{x=\pm b} =0.
\label{eq:u0-BC-1}
\end{equation}

From \eqref{eq:u0-eqn} and \eqref{eq:u0-BC-1} arise several distinct linear eigenvalue problems for $H_V=-\partial_x^2+V$. Those which are associated with the zeros of $w(k)$ and $s_\pm(k)$, are as follows:

\subsubsection*{I. Outgoing radiation problem for bound states or scattering resonance states}

\begin{subequations}
\label{eq:outgoing}
\begin{align}\label{eq:outgoing+b}
    u_0'  - ik_0 u_0 &= 0\quad \textrm{at $x=+b$}\\
     u_0'  + ik_0 u_0 &= 0\quad \textrm{at $x=-b$}.
     \label{eq:outgoing-b}
\end{align}
\end{subequations}

Solutions of the outgoing radiation problem correspond to $k_\star\in \mathbb{C}_+\cup \mathbb{C}_-\cup \{0\}$ at which $w(k_\star)=0$. 
Recall that if $w(k_\star)=0$ with $\Im k_\star>0$, then $k_\star=i\kappa_\star$, with $\kappa_\star>0$ and 
 $E_0=k_\star^2=-\kappa^2$ is an eigenvalue of $H_V$ acting in $L^2(\mathbb R)$.
  Further, if $w(k_\star)=0$ with $\Im k_\star<0$ with scattering resonance pair $(k_0,\psi)$, then $(-\overline{k_0},\overline\psi)$ is also a scattering resonance state.

\subsubsection*{II. Right transmission problem}

\begin{subequations}
\label{eq:r-transmission}
\begin{align}\label{eq:r-transmission+b}
    u_0'  - ik_0 u_0 &= 0\quad \textrm{at $x=+b$}\\
     u_0'  - ik_0 u_0 &= 0\quad \textrm{at $x=-b$}.
     \label{eq:r-transmission-b}
\end{align}
\end{subequations}

Solutions arise from the zeros, $k\in \mathbb C\setminus\{0\}$, of $s_-(k)$. By symmetry, solutions of the incoming radiation and left transmission problems arise, respectively, from the zeros $w(-k)$ and $s_+(k)=s_-(-k)$.

\subsubsection*{Potential for bifurcations arising from the zeros of $w(k)$ and $s_-(k)$}

The above discussion suggests that the zeros of $w(k)$ and $s_-(k)$  potentially give rise, for $\varepsilon$ small, to nontrivial 
formal power series solutions \eqref{eq:power_series} of nonlinear boundary value problem  \eqref{eq:nonlinear-bvp}.
For these, in turn, to generate time-harmonic nonlinear bound states of NLS / GP, we require that their energy $E=k_\star^2$ be real. Thus,
\[\textrm{\it we restrict attention to  zeros, $k_\star=i\kappa_\star$, of $k\mapsto w(k)$ and $k\mapsto s_-(k)$ located on the imaginary axis.}\]
\subsection*{Bifurcation theorem} Let $k_\star=i\kappa_\star\in\mathbb C$, with $\kappa_\star\in\mathbb R\setminus\{0\}$, denote a simple \underline{non-zero} and purely imaginary solution of  $w(k)=0$ or $s_-(k)=0$. Assume that $U_\star(x)$ denotes a choice of corresponding solution of $H_VU_\star=k_\star^2 U_\star =-\kappa_\star^2U_\star$ for the outgoing radiation problem, where $w(k_\star)=0$ and,  for the transmission resonance problem either $s_-(k_\star)=0$ or $s_+(k_\star)=0$. In the case where $\kappa_\star=\Im k_\star<0$ (scattering resonance pole)
we assume the non-degeneracy condition:
\begin{equation} 2\kappa_\star \int_{-b}^b U_\star^2(x) dx + U^2_\star(-b) + U^2_\star(+b)\ \ne\ 0.\label{eq:non-degen-scat-pole}
\end{equation}

In the case where $\kappa_\star$ is a transmission resonance, we assume:
\begin{equation}
   2\kappa_\star \int_{-b}^b U_\star^2(x)\,dx 
   - U_\star^2(-b) \, + \, U_\star^2(+b) \;\neq\; 0.
\end{equation}

\begin{thm}[Bifurcations from Bound States and Scattering/Transmission Resonances]\label{thm:bifurcations} 

Assume the above setup, where $k_\star=i\kappa_\star\ne0$ is a simple purely imaginary and non-zero bound state pole, scattering resonance pole or transmission resonance. Equivalently, we assume that $k=i\kappa_\star$ is a simple zero or a simple pole of either $k\mapsto r_+(k)=s_+(k)/w(k)$ or $k\mapsto r_-(k)=s_-(k)/w(k)$, the reflection coefficients of $H_V$.\\
 Then, there exists $\varepsilon_0>0$ such that for all $0<\varepsilon<\varepsilon_0$
 the following holds:
\begin{enumerate}
    \item There is a solution $\big(E(\varepsilon),u(x,\varepsilon)\big)$
     of the nonlinear boundary value problem \eqref{eq:nonlinear-bvp}, such that  $E(\varepsilon):(0,\varepsilon_0)\to(-\infty,0)$ is analytic, where $E(0)=-\kappa_\star^2$ and $u(x,0)=U_\star(x)$,  for each $x\in(-b,b)$ and 
     \begin{subequations}
     \begin{align}
            E(\varepsilon) &= -\kappa_\star^2 + \mathcal{O}(\varepsilon)\\
            u(x,\varepsilon) & = 
                 U_\star(x) + \mathcal{O}_{L^\infty}(\varepsilon),
                \quad  -b \le x\le +b.
        \end{align}
        \end{subequations}

     \item In all cases in part (1), the solution can be continued to all $x\in\mathbb R$ as a nonlinear bound state of NLS/GP.  There exist translates $x_{\rm R}(\varepsilon)$ and $ 
     x_{\rm L}(\varepsilon)\in\mathbb R$, such that NLS/GP \eqref{eq:V-nls-full} has a time-harmonic nonlinear bound state solution defined on all $\mathbb R$: $\Psi(x,t) = e^{-iE(\varepsilon)t} \psi(x,\varepsilon)$, where 
        \begin{align}\label{eq:nl-bound-state}
            E(\varepsilon) &= -\kappa_\star^2 + \mathcal{O}(\varepsilon)\\
            \psi(x,\varepsilon) & = \begin{cases}
             \pm \mathcal{S}\big(x-x_{\rm L}(\varepsilon), E(\varepsilon)\big), & x<-b\\
                \sqrt{\varepsilon}\Big(\ U_\star(x) + \mathcal{O}_{L^\infty}(\varepsilon) \Big),
                & -b \le x\le +b\\
                \pm \mathcal{S}\big(x-x_{\rm R}(\varepsilon), E(\varepsilon)\big), & x>+b\ .
            \end{cases}
        \end{align}

        Here, $\mathcal{S}(x,E)$ denotes the 1-soliton; see Theorem \ref{thm:1soliton}.
        \item  For $\varepsilon\to0$, if $k_\star$ is a bound state or scattering resonance pole,  then 
         \begin{subequations}
 \label{eq:xRL-1}
 
\begin{align}
x_{\rm R}(\varepsilon) 
&\approx \Bigg[\, b + \frac{1}{\kappa_\star}
   \log\!\Big(\tfrac{|U_\star(b)|}{2\sqrt{2}\,\kappa_\star}\Big) \,\Bigg]
   - \frac{1}{2\kappa_\star}\,\log \! \left( \frac 1 \varepsilon \right), \\[6pt]
x_{\rm L}(\varepsilon) 
&\approx \Bigg[\, -b - \frac{1}{\kappa_\star}
   \log\!\Big(\tfrac{|U_\star(-b)|}{2\sqrt{2}\,\kappa_\star}\Big) \,\Bigg]
   + \frac{1}{2\kappa_\star}\,\log \! \left( \frac 1 \varepsilon \right).
\end{align}
  \end{subequations}

  If $k_\star$ is a transmission resonance, then
           \begin{subequations}
 \label{eq:xRL-2}
\begin{align}
x_{\rm R}(\varepsilon) 
&\approx \Bigg[\, b + \frac{1}{\kappa_\star}
   \log\!\Big(\tfrac{|U_\star(b)|}{2\sqrt{2}\,\kappa_\star}\Big) \,\Bigg]
   - \frac{1}{2\kappa_\star}\,\log \! \left( \frac 1 \varepsilon \right), \\[6pt]
x_{\rm L}(\varepsilon) 
&\approx \Bigg[\, -b + \frac{1}{\kappa_\star}
   \log\!\Big(\tfrac{|U_\star(-b)|}{2\sqrt{2}\,\kappa_\star}\Big) \,\Bigg]
   - \frac{1}{2\kappa_\star}\,\log \! \left( \frac 1 \varepsilon \right).
\end{align}

  \end{subequations}
\end{enumerate}
\end{thm}

Figures~\ref{fig:boundstate}, \ref{fig:resonance}, and \ref{fig:transmission}  are schematics which illustrate the character of these bifurcation scenarios. \\

\begin{rem}[Generalization to focusing nonlinearities]
The conclusions of Theorem~\ref{thm:bifurcations} extend to general focusing nonlinearities $-f(|\psi|^2)$ that admit a homoclinic orbit (localized soliton) via an analogous gluing and bifurcation analysis.
\end{rem}

\medskip

\subsubsection{Proof of part 1 of Theorem \ref{thm:bifurcations}}
  We seek solutions of the nonlinear boundary value problem
  \begin{subequations}
    \begin{align}\label{eq:nonlinear-bvp1}
&\Big(-\partial_x^2 + V(x) - \varepsilon\ u(x)^2\Big)u(x) = E u(x),\quad \lvert x \rvert <b\\
&   \mathfrak{E}[u](\pm b)\ =\ 0, \quad {\rm where}\\
&   \mathfrak{E}[u]\equiv  \big(u'\big)^2  + \frac{\varepsilon}{2} u^4 + E u^2.
\end{align}
\end{subequations}

We focus at first on the case of states arising from bound state poles of $H_V$: $w(k_\star) = 0$, $k_\star=i\kappa_\star$, with $\kappa_\star > 0$. The corresponding eigenpair of $H_V$ is denoted $(E_\star=-\kappa_\star^2, U_\star(x))$. We later turn to  the case of scattering resonances and transmission zeros.

Let $E=k^2=(i\kappa)^2=-\kappa^2$, with $\kappa > 0$.  
We shoot from $x=-b$ to $x=b$.  For  all $\kappa$ in a real neighborhood of $\kappa_\star$, and $\varepsilon$ sufficiently small,  define  $U=U(x,\kappa,\varepsilon)$ to be the unique solution of the initial value problem with initial conditions imposed at $x=-b$:
\begin{subequations}
\label{eq:U-IVP}
\begin{align}
  &  \Big(-\partial_x^2 + V(x) - \varepsilon\  U^2(x,\kappa,\varepsilon)\Big)U(x,\kappa,\varepsilon)= -\kappa^2 U(x,\kappa,\varepsilon) \label{eq:U-eqn}\\
 &   U(-b,\kappa,\varepsilon) =1\\
 & \partial_x U(-b,\kappa,\varepsilon) - U(-b,\kappa,\varepsilon)\sqrt{ \kappa^2  - \frac{\varepsilon}{2}U^2(-b,\kappa,\varepsilon)}\ =\ 0 . \label{eq:rad-minus-b}
\end{align}
\end{subequations}
For fixed $\varepsilon$, $\big(\kappa, U(x,\kappa,\varepsilon)\big)$ is a solution  of the nonlinear boundary value problem
\eqref{eq:nonlinear-bvp1} if $U$ satisfies the outgoing boundary condition at $x=+b$:
\begin{equation}
F(\kappa,\varepsilon)\ \equiv \ \partial_x U(+b,\kappa,\varepsilon) +  U(+b,\kappa,\varepsilon)\ \sqrt{ \kappa^2  - \frac{\varepsilon}{2} U^2(+b,\kappa,\varepsilon)} \ =\ 0.
\label{eq:nlbc-at+b}
\end{equation}

The map $(\kappa,\varepsilon)\mapsto F(\kappa,\varepsilon)$ is smooth for $|\kappa-\kappa_\star|\, <\kappa_1$ and $|\varepsilon|\, <\varepsilon(\kappa_1)$ sufficiently small.

For $\varepsilon=0$ we have $F(\kappa,0)= \partial_x U(+b,\kappa,0) + \kappa U(+b,\kappa,0)$. By hypothesis, there is a non-trivial solution of the boundary value problem
\begin{align*} 
&\Big(-\partial_x^2 + V(x)\Big)U_\star(x) = -\kappa_\star^2\ U_\star(x)\\
&\partial_x U_\star(-b) - \kappa_\star U_\star(-b) = 0\\
&\partial_x U_\star(+b) + \kappa_\star U_\star(+b) = 0.
\end{align*}
Since $U_\star$ is non-trivial, $U_\star(-b)\ne0$, and  so we can divide $U_\star(x)$ by $U_\star(-b)$  to arrange for  $U_\star(-b)=1$. With this normalization we have, by uniqueness, that $U(x,\kappa_\star,0)=U_\star(x)$.  Hence,  $F(\kappa_\star,0)=0$.

We next apply the implicit function theorem to infer the existence of a curve $\varepsilon\mapsto \kappa(\varepsilon)$, defined in a neighborhood of $\varepsilon=0$, such that $F(\kappa(\varepsilon),\varepsilon)=0$. It suffices to verify that $\partial_\kappa F(\kappa_\star,0)=0$.

Let \[\dot{U}_\star(x) = \partial_\kappa U(x,\kappa,\varepsilon)\Big|_{(\kappa,\varepsilon)=(\kappa_\star,0)}.\]
Differentiating the expression for $F(\kappa,\varepsilon)$ in \eqref{eq:nlbc-at+b} with respect to $\kappa$ and setting $(\kappa,\varepsilon)=(\kappa_\star,0)$, we find that the condition to be verified is:
\begin{equation}
\partial_\kappa F(\kappa_\star,0)\ =\ \partial_x \dot{U}_\star(b) + \kappa_\star \dot{U}_\star(b) + U_\star(b)\ \ne\ 0.
\label{eq:ift-cond}\end{equation}
Hence our next task is to evaluate the expression in \eqref{eq:ift-cond}.

An equation for $\dot{U}(x)$ is obtained from differentiation of \eqref{eq:U-eqn} with respect to $\kappa$ and setting $(\kappa,\varepsilon)=(\kappa_\star,0)$:
\[
 -\dot{U}_\star'' +V\dot{U}_\star +\kappa_\star^2\dot{U}_\star = -2\kappa_\star U_\star.
\]
Multiplying by $U_\star$ we obtain, using the equation for $U_\star$,
\[ \big(U_\star' \dot{U}_\star\big)' - \big( U_\star \dot{U}_\star'\big)' = - 2\kappa_\star U_\star^2.\]
Then, integrating over $-b\le x\le b$ we obtain:
\[
\Big(U_\star'(b) \dot{U}_\star(b) - U_\star'(-b) \dot{U}_\star(-b)\Big)\ -\ 
\Big( U_\star(b) \dot{U}_\star'(b) - U_\star(-b) \dot{U}_\star'(-b) \Big) \ =\ - 2\kappa_\star \int_{-b}^b U_\star^2(x) dx.
\]
Using the boundary conditions of $U_\star$: $U_\star'(-b) = \kappa_\star U_\star(-b)$
 and $U_\star'(+b) = -\kappa_\star U_\star(b)$ we have
\[
\Big(-\kappa_\star U_\star(b) \dot{U}_\star(b) - \kappa_\star U_\star(-b) \dot{U}_\star(-b)\Big)\ -\ 
\Big( U_\star(b) \dot{U}_\star'(b) - U_\star(-b) \dot{U}_\star'(-b) \Big) \ =\ - 2\kappa_\star \int_{-b}^b U_\star^2(x) dx.
\]
Therefore, 
\[
U_\star(b)\Big(\dot{U}_\star'(b)+\kappa_\star \dot{U}_\star(b) \Big) = 
2\kappa_\star \int_{-b}^b U_\star^2(x) dx + 
U_\star(-b)\Big(\dot{U}_\star'(-b)-\kappa_\star \dot{U}_\star(-b) \Big).
\]
Next, we can eliminate the $\dot{U}_\star$ dependence by differentiating the boundary condition of $U$ at $x=-b$ with respect to $\kappa$ and setting $(\kappa,\varepsilon)=(\kappa_\star,0)$, giving
\[\ \dot{U}_\star'(-b) - \kappa_\star \dot{U}_\star(-b) = U_\star(-b) .\]
Hence,
\[
U_\star(b)\Big(\dot{U}_\star'(b)+\kappa_\star \dot{U}_\star(b) \Big) = 
2\kappa_\star \int_{-b}^b U_\star^2(x) dx + 
U^2_\star(-b).
\]
Therefore,
\begin{equation}
U_\star(b)\ \partial_\kappa F(\kappa_\star,0) = U_\star(b)\Big( \dot{U}'_\star(b)  +\kappa_\star \dot{U}_\star(b) \Big) + U^2_\star(b)
 = 2\kappa_\star \int_{-b}^b U_\star^2(x) dx + U^2_\star(-b) + U^2_\star(b) 
\label{eq:Fstar}
\end{equation}
or
\begin{equation}\label{eq:parF}
\partial_\kappa F(\kappa_\star,0) = 
\frac{2\kappa_\star \int_{-b}^b U_\star^2(x) dx + U^2_\star(-b) + U^2_\star(b)}{U_\star(b)}.
\end{equation}
Again, since $U_\star$ is non-trivial, it must be that $U_\star(b)\ne0$ and we obtain the  expression \eqref{eq:parF} for  $\partial_\kappa F(\kappa_\star,0)$, whose non-vanishing or vanishing we may investigate.

Note that if $k_\star=i\kappa_\star$ is a bound state pole, then $\kappa_\star>0$. In this case, $\partial_\kappa F(\kappa_\star,0)$ is non-zero. The implicit function theorem then ensures a bifurcation. This completes the proof of part (1) for the case of a bound state pole.

For the analysis of a scattering resonance: $k_\star=i\kappa_\star$ with $\kappa_\star<0$,  we proceed analogously. In this case the required boundary conditions are:
\begin{equation}\label{eq:sys-res}
\left\{
\begin{aligned}
   &\partial_x U(-b,\kappa,\varepsilon) 
     + U(-b,\kappa,\varepsilon)
       \sqrt{\kappa^2 - \tfrac{\varepsilon}{2}U^2(-b,\kappa,\varepsilon)} = 0, \\[0.4em]
   &\partial_x U(b,\kappa,\varepsilon) 
     - U(b,\kappa,\varepsilon)
       \sqrt{\kappa^2 - \tfrac{\varepsilon}{2}U^2(b,\kappa,\varepsilon)} = 0 .
\end{aligned}
\right.
\end{equation}

An analogous calculation yields the expression in \eqref{eq:parF} for  $\partial_k F$, although now with $\kappa_\star<0$. To apply the implicit function theorem, ensuring the existence of a bifurcation, we  impose the non-degeneracy condition on the scattering resonance pair $(i\kappa_\star,U_\star)$:
 \begin{equation}\label{eq:non-deg} 2\kappa_\star \int_{-b}^b U_\star^2(x) dx + U^2_\star(-b) + U^2_\star(+b)\ \ne\ 0.\end{equation}
 This completes the proof of part (1) for the case of a scattering resonance pole.
  
Finally, we turn to bifurcations from transmission resonances. These correspond to the zeros of $s_\pm(k)$ and, without loss of generality, we restrict to the zeros of $s_-(k)$.
As with the case of scattering resonances, we consider transmission resonances located on the imaginary axis. If $s_-(i\kappa_\star)=0$, with $\kappa_\star>0$, then the corresponding transmission resonance mode decays as $x\to+\infty$ and grows as $x\to-\infty$. Transmission resonance modes for $k_\star$ on the negative imaginary axis decay as $x\to-\infty$ and grow as $x\to+\infty$. The appropriate boundary conditions for nonlinear bound states bifurcating from  transmission resonances are therefore:
\begin{equation}\label{eq:tr-bc}
\left\{
\begin{aligned}
   &\partial_x U(-b,\kappa,\varepsilon) 
      \pm U(-b,\kappa,\varepsilon)
      \sqrt{\kappa^2 - \tfrac{\varepsilon}{2}U^2(-b,\kappa,\varepsilon)} = 0, \\[0.4em]
   &\partial_x U(b,\kappa,\varepsilon) 
      \pm U(b,\kappa,\varepsilon)
      \sqrt{\kappa^2 - \tfrac{\varepsilon}{2}U^2(b,\kappa,\varepsilon)} = 0,
\end{aligned}
\right.
\end{equation}
where the signs $+$ or $-$ corresponding, respectively, to whether $k_\star\in\mathbb C_+$ or $\C_-$.

An analogous calculation to the case of bound state poles and scattering resonance poles yields, 
for the case of transmission resonances:
\begin{equation}\label{eq:tr-cond}
\partial_\kappa F(\kappa_\star,0) 
= \frac{2\kappa_\star \int_{-b}^b U_\star^2(x)\,dx 
      - U_\star^2(-b) + U_\star^2(+b)}{U_\star(+b)}.
\end{equation}
 Therefore, under the assumption that
 \[ 2\kappa_\star \int_{-b}^b U_\star^2(x)\,dx 
      - U_\star^2(-b) + U_\star^2(+b)\ne0, \]
      the implicit function theorem implies the 
      existence of a bifurcation. This concludes
      the proof of Part (1) for the case of transmission resonances. The proof of Part 1 of Theorem \ref{thm:bifurcations} is now complete.

  \subsubsection*{Proof of Part 2 of Theorem \ref{thm:bifurcations}} \label{sec:part2} We next extend  the solution 
   $\big(\kappa(\varepsilon), U(x,\kappa(\varepsilon),\varepsilon)\big)$ of the boundary value problem in Part 1 to a nonlinear bound state pair $(-\kappa(\varepsilon)^2,\psi)$ of  NLS/GP \eqref{eq:V-nls-full} on all $\mathbb R$: 
    \begin{subequations}
  \label{eq:inner-outer-psi}
  \begin{align}
      \Big(-\partial_x^2 + V - \psi^2\Big)\psi &= -\kappa(\varepsilon)^2  \psi,\quad \lvert x \rvert  < b \label{eq:inner}\\
       \Big(-\partial_x^2 - \psi^2\Big) \psi &= -\kappa(\varepsilon)^2  \psi,\quad \lvert x \rvert > b \label{eq:outer}
  \end{align}
  \end{subequations}
  together with continuity of $\psi(x)$ and $\psi'(x)$ at $x=\pm b$.

 Since $\psi=\eta u$ and $\varepsilon=\eta^2$,  (see \eqref{eq:psietau}) the function $\psi(x,\varepsilon) = \sqrt{\varepsilon} U(x,\kappa(\varepsilon),\varepsilon)$ is a solution of the nonlinear bound state  problem, \eqref{eq:inner}, on $\lvert x \rvert<b$. To continue $\psi(x,\varepsilon)$ to a solution of \eqref{eq:inner-outer-psi} on $\mathbb R$, must impose continuity conditions on $\psi(x)$ and $\psi'(x)$ at $x=\pm b$.
\subsection*{\texorpdfstring{Continuity of $\psi(x)$ at $x = +b$}{Continuity of psi(x) at x = +b}}
This is equivalent to
\[
\sqrt{\varepsilon}\,U\big(b,\kappa(\varepsilon),\varepsilon\big)
= \operatorname{sgn}\!\Big(U\big(b,\kappa(\varepsilon),\varepsilon)\Big)\,
\mathcal S\Big(b-x_{\rm R}(\varepsilon),-\kappa(\varepsilon)^2\Big),
\]
or, in magnitude,
\[
\sqrt{\varepsilon}\,\big|U\big(b,\kappa(\varepsilon),\varepsilon)\big|
= \mathcal S\Big(b-x_{\rm R}(\varepsilon),-\kappa(\varepsilon)^2\Big).
\]
For $\varepsilon>0$ small with $\kappa(\varepsilon)\to\kappa_\star>0$,
\[
\sqrt{\varepsilon}\,\big|U_\star(b)\big|
= \mathcal S\big(b-x_{\rm R}(\varepsilon),-\kappa_\star^2\big).
\]

Then taking $\varepsilon$ small enough and $\kappa_\star \neq 0$, the left hand side is in the range of the right-hand side, thus there is a solution $x_{\rm R}$.

Specifically, assume that $\kappa_\star>0$ (bound state pole or scattering resonance pole). Then,   $|U_\star|$ is \emph{decreasing} at $x=b$, requiring a matching to the right flank of $\mathcal{S}$. Thus,   $b-x_{\rm R}>0$ and $x_{\rm R}(\varepsilon)$ is given by, for $\varepsilon$ small by:
\[
\sqrt{\varepsilon}\,\big|U_\star(b)\big|\;\approx\;
2 \sqrt{2 \kappa_\star^2} \exp\!\big(-\kappa_\star(b-x_{\rm R})\big).
\]
Hence, 
\begin{equation}
\begin{aligned}
x_{\rm R}(\varepsilon) 
&\approx b+\frac{1}{\kappa_\star}
   \log\!\Big(\frac{\sqrt{\varepsilon}\,\big|U_\star(b)\big|}{2 \sqrt{2 \kappa_\star^2}}\Big) \\[6pt]
&= \Bigg[\, b + \frac{1}{\kappa_\star}
   \log\!\Big(\tfrac{|U_\star(b)|}{2\sqrt{2}\,\kappa_\star}\Big)\,\Bigg]
   - \frac{1}{2\kappa_\star}\,\log\!\left(\tfrac{1}{\varepsilon}\right)
   \;\longrightarrow\; -\infty.
\end{aligned}
\end{equation}

If $\kappa_\star < 0$ (bound state pole or scattering resonance pole), then the profile $|U_\star|$ is \emph{increasing} at $x=b$, requiring a matching to the left flank of $\mathcal{S}$. Thus,  $b-x_{\rm R}<0$ we 
\[
\sqrt{\varepsilon}\,\big|U_\star(b)\big|\;\approx\;
2 \sqrt{2 \kappa_\star^2} \exp\!\big(-\kappa_\star(b-x_{\rm R})\big).
\]
In this case, 
\begin{equation}
\begin{aligned}
x_{\rm R}(\varepsilon) 
&\approx b+\frac{1}{\kappa_\star}
   \log\!\Big(\frac{\sqrt{\varepsilon}\,\big|U_\star(b)\big|}{2 \sqrt{2 \kappa_\star^2}}\Big) \\[6pt]
&= \Bigg[\, b + \frac{1}{\kappa_\star}
   \log\!\Big(\tfrac{|U_\star(b)|}{2\sqrt{2}\,\kappa_\star}\Big)\,\Bigg]
   - \frac{1}{2\kappa_\star}\,\log\!\left(\tfrac{1}{\varepsilon}\right)
   \;\longrightarrow\; +\infty.
\end{aligned}
\label{eq:xR}
\end{equation}

\subsection*{\texorpdfstring{Continuity of $\psi'(x)$ at $x = +b$}{Continuity of psi'(x) at x = +b}}
That $\psi'$ is continuous at $x=+b$ follows from the fact that the inner solution ($\lvert x \rvert \, \le b$), given by $\psi(x,\varepsilon) = \sqrt{\varepsilon} U(x,\kappa(\varepsilon),\varepsilon)$ and the outer solution ($x\ge b$) given by $\psi(x,\varepsilon) = {\rm sgn}\Big(U(b,\kappa(\varepsilon),\varepsilon)\Big)\times \mathcal{S}\Big(b-x_{\rm R}, -\kappa^2(\varepsilon)\Big)$ both satisfy one of the boundary conditions
\begin{equation}
\partial_x \psi(b,\varepsilon) \pm  \psi(b,\varepsilon)\ \sqrt{ \kappa^2(\varepsilon)  - \frac{\varepsilon}{2} \psi^2(b,\varepsilon)} \ =\ 0
\label{eq:psi-nlbc-at+b}
\end{equation}
at $x=+b$, and the continuity of $\psi$ at $x=+b$.  Equation 
\eqref{eq:psi-nlbc-at+b} and continuity of $\psi$ at $x=+b$ imply:
$\partial_x\psi(b^-,\varepsilon)=
 \partial_x \psi(b^+,\varepsilon)$.

 Analogous considerations establish 
 the continuity of both $\psi$ and $\psi'$ at $x=-b$ and
\begin{equation}
\begin{aligned}
x_{\rm L}(\varepsilon) 
&\approx -b \mp \frac{1}{\kappa_\star}
   \log\!\Big(\frac{\sqrt{\varepsilon}\,\big|U_\star(-b)\big|}{2 \sqrt{2 \kappa_\star^2}}\Big), \\[6pt]
&= \Bigg[\, -b \mp \frac{1}{\kappa_\star}
   \log\!\Big(\tfrac{|U_\star(-b)|}{2\sqrt{2}\,\kappa_\star}\Big)\,\Bigg]
   \mp \frac{1}{2\kappa_\star}\,\log\!\left(\tfrac{1}{\varepsilon}\right)
\end{aligned}
\label{eq:xL}
\end{equation}
where $\mp$ is chosen depending on whether $k_\star=i\kappa_\star$ is a zero of $w(k)$ or $s_-(k)$, respectively. 
By convention, transmission resonances have growth/decay at $x=-b$ opposite to bound states or scattering resonances. In other words, if $\kappa_\star>0$, then $w(i\kappa_\star)=0$ corresponds to decay and $s_-(i\kappa_\star)=0$ to growth, 
while for $\kappa_\star<0$ the roles are reversed. This choice depends only on the behavior at $x=-b$ by our convention of using $s_-(k)$.

 Now for $\varepsilon=\eta^2$ small and positive, and $\kappa_\star \neq 0$, the expression in \eqref{eq:nl-bound-state} defines a nonlinear bound state on all $\mathbb R$ in $H^1(\mathbb R)$, thus completing the proof of Part 2 of Theorem \ref{thm:bifurcations}. \\
 
   With $x_{\rm L}$ and $x_{\rm R}$ defined as in \eqref{eq:xR} and \eqref{eq:xL}, we can simplify them further into the expressions given in Part 3 \eqref{eq:xRL-1} and \eqref{eq:xRL-2}, completing the proof of Theorem~\ref{thm:bifurcations}.\\
  \hfill \qed

  \medskip

\subsection{The case of a threshold (zero energy) resonance.}\label{sec:ZER}

For one-dimensional potential wells, $H_V$ always has at least one bound-state pole $i\kappa_0$ with $\kappa_0>0$.  
As the well is made deeper, an additional bound-state pole appears when a scattering resonance (a zero of $w(k)$ in the lower half-plane) moves upward along the imaginary axis and passes through $k=0$.  
The point of transition, where $w(k_\star)=0$ with $k_\star=i\kappa_\star=0$, corresponds to a \emph{threshold resonance}.  
At this point, the scattering resonance and the transmission resonance problems ``degenerate''; they both impose Neumann boundary conditions at $x=\pm b$.

This scenario is illustrated for various cases in Section~\ref{sec:zen-res}. In all cases shown, scattering and transmission resonances coalesce at $k=0$ as the potential depth is varied.  
The resulting threshold resonance gives rise to multiple nonlinear branches bifurcating from $(E,\mathcal N)=(0,0)$ in the cases of Figures~\ref{fig:zerobif2} and \ref{fig:zerobif3}.
 
It turns out that the analytical arguments  which we applied to produce bifurcations from non-zero
bound state and resonance poles and transmission resonances (Theorem \ref{thm:bifurcations}),
 do not immediately extend to the case where $k_\star=0$; indeed $F(\kappa,\varepsilon)$ is not regular in a neighborhood of $(\kappa_\star,\varepsilon)=(0,0)$.  
In the following, we therefore restrict our study to the bifurcations having  even or odd symmetry; the boundary conditions of transmission resonances preclude these symmetries. A study of the general case, where multiple bifurcations may occur, is work in progress.

\subsubsection{\bf Bifurcation from an even or odd threshold scattering resonance pole}\label{sec:bifthr}
\begin{lem}\label{lem:evenorodd}
    If $H_V$ has a threshold resonance $H_V U_\star = 0$ such that $U_\star'(\pm b) = 0$ and $U_\star \not \equiv 0$, then it is unique up to a non-zero constant. Furthermore, if $V(x) = V(-x)$, then $U_\star$ is either odd or even.
\end{lem}
\begin{proof}
   Let $V(x) = V(-x)$ and $u_1$ be a nontrivial solution of
$ H_V u_1 = 0$ with boundary conditions $u_1'(\pm b) = 0$.
By the symmetry of $V$, $u_2(x) = u_1(-x)$ is also a solution. However, their Wronskian vanishes and hence they are linearly dependent; $u_1(x) = c u_1(-x) = c^2 u_1(x)$ for $c \neq 0$. Hence  $c$ must be either $+1$ or $-1$. In other words,  $u_1$ is either symmetric or anti-symmetric.
\end{proof}

We now state our result on bifurcation from a threshold resonance at $E_\star=0$. In contrast to Theorem \ref{thm:bifurcations}, in this case the nonlinear bound state centerings, $x_{\rm L}$ and $x_{\rm R}$,  do not drift to spatial infinity as $E\to 0$.

\begin{thm}[Bifurcation from a threshold resonance under symmetry constraints]\label{thm:bifurcations-zero}
Let $V(x)$ denote an even and bounded potential. Suppose that $H_V=-\partial_x^2 + V(x)$ has a zero energy (threshold) resonance with corresponding mode $U_\star$:
\begin{align*}
H_V U_\star &= 0, \qquad 0<x<b,
\end{align*}
which satisfies the boundary conditions
(i) $U'_\star(0)=0$ and  $U_\star'(b)=0$ in the case where $U_\star$ is even, and 
  (ii) $U_\star(0)=0$ and $U_\star'(b)=0$ in the case where $U_\star$ is odd. 

Then,  there exists $\varepsilon_0>0$ such that for all $0<\varepsilon<\varepsilon_0$ the following holds:
\begin{enumerate}
\item There exists a solution $\big(E(\varepsilon),U(x,E(\varepsilon),\varepsilon)\big)$ of the nonlinear boundary value problem \eqref{eq:nonlinear-bvp} on $[-b,b]$, such that 
\[
E(\varepsilon):(0,\varepsilon_0)\to(-\infty,0)
\]
is analytic, and 
\begin{subequations}
\begin{align}
   E(\varepsilon) &= -\tfrac12 U_\star^2(b)\,\varepsilon + \mathcal O(\varepsilon^2), \label{eq:Evepsilon}\\
   U(x,E(\varepsilon),\varepsilon) &= U_\star(x) + \mathcal O_{L^\infty}(\varepsilon),\quad \lvert x \rvert <b.
\end{align}
\end{subequations}
The nonlinear bound state $U(x,E(\varepsilon),\varepsilon)$ inherits the symmetry type (even or odd)
of the scattering resonance mode $U_\star(x)$.

\item The solution in (1) of the boundary value problem \eqref{eq:nonlinear-bvp_full} can be extended to all $\mathbb R$ as a nonlinear bound state of NLS/GP defined as
$
\Psi(x,t) = e^{-iE(\varepsilon)t}\,\psi(x,\varepsilon)$, with:
\begin{equation}
\psi(x,\varepsilon) =
\begin{cases}
\sgn (U_\star(-b)) \cdot \mathcal S\big(x-x_{\rm L}(\varepsilon),E(\varepsilon)\big), & x < -b,\\
\sqrt{\varepsilon}\,\Big( U_\star(x) + \mathcal O(\varepsilon)\Big), & -b \le x \le b, \\[0.5em]
\sgn (U_\star(b)) \cdot \mathcal S\big(x-x_{\rm R}(\varepsilon),E(\varepsilon)\big), & x>b.
\end{cases}
\end{equation}
Here, $\mathcal S$ is the 1-soliton defined in Theorem~\ref{thm:1soliton}, and 
$x_{\rm R}(\varepsilon)=-x_{\rm L}(\varepsilon)$ is  chosen to ensure continuity of $\psi$ and $\psi'$ at $x=b$ (and hence at $x=-b$).

\item  The translation parameters,  $x_{\rm L}(\varepsilon)$ and $x_{\rm R}(\varepsilon)$,  have finite limits as $\varepsilon \downarrow 0$: 
\begin{subequations}\label{eq:xRL-zero}
\begin{align}
x_{\rm R}(\varepsilon) &= b + \frac{1}{U_\star^2(b)} 
 \left( \int_0^b U_\star^2(x)\,dx \;-\; \frac{2}{U_\star^2(b)} \int_0^b U_\star^4(x)\,dx \right) + \mathcal O(\varepsilon), \\
 x_{\rm L}(\varepsilon) &= - x_{\rm R}(\varepsilon).
\end{align}
\end{subequations}
\end{enumerate}
\end{thm}

\subsubsection{Proof of part 1 of Theorem~\ref{thm:bifurcations-zero}}

Let $U_\star$ denote a zero energy threshold resonance of $H_V$. By Lemma \ref{lem:evenorodd},  $U_\star$ is either even or odd. Define $U(x;E,\varepsilon)$ to be the unique solution 
 of $H_VU-\varepsilon U^3=EU$ on $0\le x\le b$ with initial condition at $x=0$
\begin{equation}
\begin{cases}
U(0;E,\varepsilon) =  1,\quad \partial_x U(0;E,\varepsilon) =  0, & \textrm{if $U_\star$ is odd}\\
 U(0;E,\varepsilon) =  0,\quad \partial_xU(0;E,\varepsilon) =  1, & \textrm{if $U_\star$ is even}
    \end{cases}
\end{equation}
By uniqueness, the corresponding solution $U(x;E,\varepsilon)$ is respectively odd or even.
 Consider the smooth map \[F : (\varepsilon,E) \mapsto U'(b; E, \varepsilon)^2 + E U(b; E, \varepsilon)^2 - \frac \varepsilon 2 U^4(b;E,\varepsilon)^4.\] 
If  $(E(\varepsilon),\varepsilon)$ is such that $F(E,\varepsilon) =0$,  then $U(x;E,\varepsilon)$  satisfies \eqref{eq:nonlinear-bvp_BC} at $x=b$. Further, since \eqref{eq:nonlinear-bvp_BC} is invariant under $x\mapsto-x$ and $U\mapsto -U$, it would then follow that $U(x,E,\varepsilon)$ satisfies this boundary condition  at $x=-b$ as well. Further, conjugate symmetry and uniqueness imply that the solution $(E(\varepsilon), U(x;E,\varepsilon))$ is real. 

Note that $U_\star = U(x;0,0)$ and $F(0,0) = 0$. Further, for $\varepsilon > 0$,
\begin{equation}
    \partial_E F(0,0) = U_\star(b)^2 \neq 0.
\end{equation}
Therefore, by the  implicit function theorem, there is a unique curve $(E(\varepsilon),\varepsilon)$ for small $|\varepsilon|$ that satisfies $F(E(\varepsilon),\varepsilon) = 0$.

 This proves part 1 of Theorem~\ref{thm:bifurcations-zero}.
\subsubsection{Proof of part 2 of Theorem~\ref{thm:bifurcations-zero}} Let $\varepsilon >0$. The pair $E(\varepsilon), \psi(x,\varepsilon)=\sqrt{\varepsilon} U(x,E(\varepsilon),\varepsilon$ solves NLS/GP on the interval $\lvert x \rvert <b$. We seek to extend this solution to one defined on all of $\mathbb{R}$, which vanishes as $\lvert x \rvert\to\infty$. We proceed as in Section~\ref{sec:part2}. For $x>b$, $\psi = \sgn(\psi(b)) \cdot \mathcal S(x-x_{\rm R};E)$, and we need to impose continuity of the solution and its first derivative at $x=b$. $x_{\rm L} = -x_{\rm R}$ will be obtained via symmetry.

We will choose $x_{\rm R}$ to satisfy continuity of the logarithmic derivative $\psi'/\psi$ across $x=b$. We will then use this $x_{\rm R}$ to verify continuity of $\psi$ and $\psi'$.

Since $ \psi(x,\varepsilon)=\sqrt{\varepsilon} U(x,E(\varepsilon),\varepsilon$, we have 
\begin{equation} \label{eq:equiv-log}
     \left(\frac{U'(b;E(\varepsilon),\varepsilon)}{U(b;E(\varepsilon),\varepsilon)} \right)^2 
    =  \left(\frac{\psi'(b;E(\varepsilon),\varepsilon)}{\psi(b;E(\varepsilon),\varepsilon)} \right)^2.
\end{equation}

Furthermore, for small $\varepsilon > 0$, the condition $F(E(\varepsilon),\varepsilon) = 0$ implies
\begin{equation}\label{eq:log-deriv}
     0 \;\leq\; \left(\frac{U'(b;E(\varepsilon),\varepsilon)}{U(b;E(\varepsilon),\varepsilon)} \right)^2
     = - E(\varepsilon) - \frac \varepsilon 2 U^2(b; E(\varepsilon), \varepsilon)
     \;<\; - E(\varepsilon).
\end{equation}

An explicit computation yields
\begin{equation}
    \frac{\mathcal S'(b-x_{\rm R};E(\varepsilon))}{\mathcal S(b-x_{\rm R};E(\varepsilon))} 
    = -\sqrt{|E(\varepsilon)|}\,\tanh\!\big( - \sqrt{|E(\varepsilon)|}\,( b - x_{\rm R} ) \big)
\end{equation}
which implies that
\begin{equation}
    \left\{ \frac{\mathcal S'(b-x_{\rm R};E(\varepsilon))}{\mathcal S(b-x_{\rm R};E(\varepsilon))} : x_{\rm R} \in \mathbb R\right\} = \left(-\sqrt{|E(\varepsilon)|} \, , \,  \sqrt{|E(\varepsilon)} \right).
\end{equation}
Furthermore, \eqref{eq:equiv-log} and \eqref{eq:log-deriv} give
\[ \left(\frac{\psi'(b;E(\varepsilon),\varepsilon)}{\psi(b;E(\varepsilon),\varepsilon)} \right) 
     \in \left(-\sqrt{|E(\varepsilon)|} \, , \,  \sqrt{|E(\varepsilon)} \right).
     \]
Hence
\[
\left(\frac{\psi'(b;E(\varepsilon),\varepsilon)}{\psi(b;E(\varepsilon),\varepsilon)} \right) 
     \in \left\{ \frac{\mathcal S'(b-x_{\rm R};E(\varepsilon))}{\mathcal S(b-x_{\rm R};E(\varepsilon))} : x_{\rm R} \in \mathbb R\right\} .
\]
Next,  select $x_{\rm R}$ to satisfy continuity of the logarithmic derivative,
\begin{equation}\label{eq:log-deriv2}
    \frac{\psi'(b;E(\varepsilon),\varepsilon)}{\psi(b;E(\varepsilon),\varepsilon)}  
    = \frac{\mathcal S'(b-x_{\rm R};E(\varepsilon))}{\mathcal S(b-x_{\rm R};E(\varepsilon))}  .
\end{equation}

Now, both $\phi(x) = \psi(x;E(\varepsilon),\varepsilon)$ and $\mathcal S(x-x_{\rm R},E(\varepsilon))$ satisfy
 \begin{equation}
     \phi'(b)^2 = - E \phi(b)^2 - \frac 12 \phi(b)^4.
 \end{equation}
So by \eqref{eq:log-deriv2},
     \begin{equation}
          \left(\frac{\psi'(b;E(\varepsilon),\varepsilon)}{\psi(b;E(\varepsilon),\varepsilon)} \right)^2 = - E(\varepsilon) - \frac 12 \psi^2(b;E(\varepsilon),\varepsilon) = -E(\varepsilon) - \frac 12 \mathcal S^2(b - x_{\rm R}; E(\varepsilon) ).
     \end{equation}
 This implies $\psi^2(b;E(\varepsilon),\varepsilon) = \mathcal S^2(b - x_{\rm R}; E(\varepsilon))$. We obtain continuity by taking $\psi(b;E(\varepsilon),\varepsilon) = \sgn( \psi(b)) \cdot \mathcal S(b - x_{\rm R}; E(\varepsilon))$ at $x=b$.  This together with
 \begin{equation}
     \frac{\psi'(b;E(\varepsilon),\varepsilon)}{\psi(b;E(\varepsilon),\varepsilon)} = \frac{\mathcal S'(b-x_{\rm R};E(\varepsilon),\varepsilon)}{\mathcal S(b-x_{\rm R};E(\varepsilon),\varepsilon)}
 \end{equation}
implies continuity of the derivative at $x=\pm b$: $\psi'(b;E(\varepsilon),\varepsilon) = \sgn(\psi(b))\cdot  \mathcal S'(b - x_{\rm R}; E(\varepsilon))$. Therefore, 
\begin{equation}
\psi(x;E(\varepsilon),\varepsilon) = \begin{cases}
    \sqrt{\varepsilon} \, U(x;E(\varepsilon),\varepsilon), & 0 \leq x \leq b \\ \sgn(\psi(b)) \cdot \mathcal S(x-x_{\rm R};E(\varepsilon)), & x > b\ .
\end{cases}
\end{equation}
$\psi(x;E(\varepsilon),\varepsilon$ can now be extended to $x<0$ to have the same symmetry as $U_\star$. This proves part 2 of Theorem~\ref{thm:bifurcations-zero}.
\subsubsection*{Proof of part 3 of Theorem~\ref{thm:bifurcations-zero}} 
Our goal is to obtain, for $\varepsilon\to0^+$, asymptotic expressions for the  translation parameters $x_{\rm R}(\varepsilon)$ and $x_{\rm L}(\varepsilon)$. Differentiation of 
 the the boundary value problem for $U$  and setting $\varepsilon$ equal to zero yields for:
$\partial_\varepsilon U\Big|_{\varepsilon=0}$:
\begin{align*}
H (\partial_\varepsilon U) &= (\partial_\varepsilon E) U_\star + U_\star^3, \\ (\partial_\varepsilon U)(0) &=(\partial_\varepsilon U)'(0)=0,\\
2 U_\star' \cdot (\partial_\varepsilon U) &= -(\partial_\varepsilon E) U_\star^2 - \frac 12 U_\star^4 \quad \textrm{at }x=b.
\end{align*}
Since $U_\star'(b) = 0$, we have
\[
 \partial_\varepsilon E(0) = - \frac 12 U_\star^2(b).
\]
By multiplying $H(\partial_\varepsilon U)$ by $U_\star$ and integrating by parts, we get the explicit boundary value
\begin{equation} \label{eq:partialU}
(\partial_\varepsilon U)'
(b;0,0) = \frac 12 U_\star(b) \int U_\star^2 \, dx - \frac{1}{U_\star(b)} \int_0^b U_\star^4(b) \, dx. 
\end{equation}
 Now we solve for $x_{\rm R}(\varepsilon)$ by continuity of logarithmic derivative:
 \begin{equation}
     \left(\frac{U'(b;E(\varepsilon),\varepsilon)}{U(b;E(\varepsilon),\varepsilon)} \right) = - \sqrt{-E(\varepsilon)} \tanh \left( \sqrt{|E(\varepsilon)|} (b - x_{\rm R}(\varepsilon) \right).
 \end{equation}
 At lowest order, this reads
 \begin{equation}
          \varepsilon \left(\frac{\partial_\varepsilon U'(b;0,0)}{U_\star(b)} \right) \approx -\sqrt{\frac \varepsilon 2} \left |U_\star(b) \right |\tanh \left( \sqrt{\frac \varepsilon 2} \left |U_\star(b) \right |(b - x_{\rm R}(\varepsilon) \right).
 \end{equation}
 Where the left-hand side simplifies by \eqref{eq:partialU} to give
 \begin{equation} \label{eq:tanh}
     \varepsilon \left( \frac 12 \int U_\star^2 \, dx - \frac{1}{U_\star^2(b)} \int_0^b U_\star^4(b) \, dx. \right) \approx -\sqrt{\frac \varepsilon 2} \left |U_\star(b) \right |\tanh \left( \sqrt{\frac \varepsilon 2} \left |U_\star(b) \right |(b - x_{\rm R}(\varepsilon) \right).
 \end{equation}

Let
\[
s \equiv \left(\frac 12  \int U_\star^2 \, dx - \frac{1}{U_\star(b)^2} \int_0^b U_\star^4(b) \, dx \right).
\]
Applying $\operatorname{arctanh}$ to both sides of \eqref{eq:tanh}, we get
\begin{equation}
    \operatorname{arctanh}\left(-\frac{\sqrt{2 \varepsilon} \, s}{|U_\star(b)|}\right) = \sqrt{\frac \varepsilon 2} \, |U_\star(b)| ( b - x_{\rm R}(\varepsilon) )
\end{equation}
or
\begin{equation}
      x_{\rm R}(\varepsilon)  \approx   b - \sqrt{\frac{2}{\varepsilon U_\star^2(b)}}\operatorname{arctanh}\left(-\frac{\sqrt{2 \varepsilon} \, s}{|U_\star(b)|}\right).
\end{equation}
Using the Taylor expansion $\operatorname{arctanh} z = z + z^3/3 + z^5/5 + \cdots$ with $z = -\frac{\sqrt{2 \varepsilon} \, s}{|U_\star(b)|} $, this becomes, when $|z| \, < 1$, 
\begin{equation}
    x_{\rm R}(\varepsilon) \approx b + \frac{2 s}{U_\star^2(b)} + \frac{4 \varepsilon s^3}{U_\star^4(b)} + \mathcal O\left( \varepsilon^2 \right).
\end{equation}
Then in the limit, the shift parameter converges to a constant:
\begin{equation}
    \lim_{\varepsilon \downarrow 0} x_{\rm R}(\varepsilon) = b \;+\; \frac{1}{U_\star^2(b)} \int_0^b U_\star^2(x)\,dx 
\;-\; \frac{2}{U_\star^4(b)} \int_0^b U_\star^4(x)\,dx
\end{equation}
and $x_{\rm L}$ is obtained via symmetry. This completes the proof of Theorem~\ref{thm:bifurcations-zero}. \\
\hfill \qed
\begin{rem}
We verify the limiting expressions for $x_{\rm L}(\varepsilon)$ and  $x_{\rm R}(\varepsilon)$, displayed in \eqref{eq:xRL-zero}, via numerical simulation for a square well with threshold resonance: $V(x) = -\frac{\pi^2}{4} \chi_{[-1,1]}(x)$. Using $U_\star(x) = \frac{2}{\pi}\sin(\pi x / 2)$, \eqref{eq:xRL-zero} reduces to  $x_{\rm R}(\varepsilon) \to 3/4$ and $x_{\rm L}(\varepsilon)=-x_{\rm R}(\varepsilon)$. The simulation gives the same result, as shown in Figure~\ref{fig:corroborate}. 
\end{rem}
\begin{figure}[!ht]
    \centering
    \includegraphics[width=0.5\linewidth]{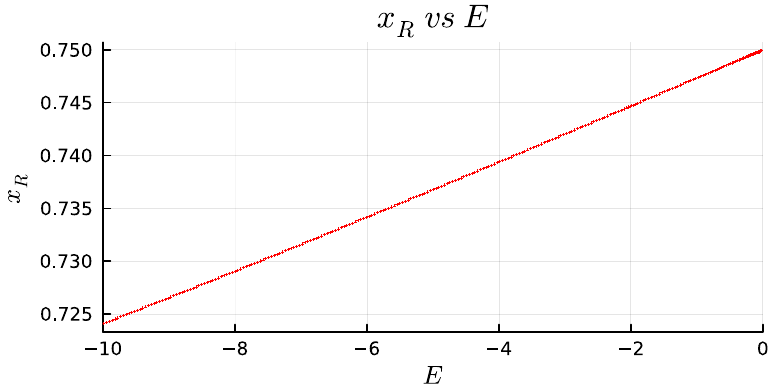}
    \caption{The calculation \eqref{eq:xRL-zero} for a square well potential $V(x) = -\frac{\pi^2}{4} \chi_{[-1,1]}(x)$ yields $x_{\rm R}(\varepsilon)\to 3/4$, for the odd nonlinear bound state  (Dirichlet condition at $x=0$) arising from an odd threshold resonance mode $U_\star(x)$. Here, we verify this numerically by computing the solutions on a fine grid of $E \in [-10,0)$ and plotting $x_{\rm R} $ vs $E$.
 }
    \label{fig:corroborate}
\end{figure}
    See also Figures~\ref{fig:zerobif1} and \ref{fig:zerobif3} where other (asymmetric) branches may bifurcate from the threshold resonance. These correspond to two transmission resonances gained when $V$ is perturbed.

\section{Bifurcation scenarios}\label{sec:bifurcation-character}

In this section we discuss the implications of Theorems~\ref{thm:bifurcations} and \ref{thm:bifurcations-zero},  in particular the  different bifurcation scenarios: 
(i) bifurcation arising from a bound state pole in subsection \ref{sec:bsbif},  (ii)
from a scattering resonance pole (anti-bound state) in subsection \ref{sec:srpbif}, and (iii) from a transmission resonance in subsection \ref{sec:trbif}. Each subsection contains a schematic figure illustrating one of these scenarios for a symmetric  
 square potential well:\\[3em]

\textbf{a) Left panel:} the locations of bound state poles ($w(k)=0$, $k\in i\mathbb R\setminus\{0\}$), scattering resonance poles ($w(k)=0$, $\Im k<0$), and transmission resonances $s_-(k)=0$, $k\in\mathbb C \setminus\{0\}$. \\[1em]

\textbf{b) Right panel:} $E\mapsto \mathcal{N}[\psi_E]$ bifurcation curve, and the plot of $x\mapsto\psi_E(x)$ corresponding to selected points along the curve.

\subsection{Bound state pole bifurcations}\label{sec:bsbif}
 Figure~\ref{fig:boundstate} describes the scenario where $H_V$ has a bound state pole at $k_\star=k_b=i\kappa_b$, with $\kappa_b>0$.  A curve of nonlinear bound states, $E\mapsto \psi_E$, emerges from  $(E,\mathcal{N})=(E_b,0)$, where $E_b=-\kappa_b^2$.  For $E$ less than and near $-\kappa_b^2$, we have $\psi_E(x)\approx \sqrt{E_b-E}\ U_\star(x)$.  This  recovers a special case of results in \cite{rose1988bound}. 

\begin{figure}[!ht]
  \centering
  % --- column titles ---
  \makebox[0.4\linewidth][c]{\textbf{Linear}}%
  \hspace{0.05\linewidth}%
  \makebox[0.4\linewidth][c]{\textbf{Nonlinear}}\\[2pt]

  % --- swapped images ---
  \includegraphics[width=0.4\linewidth,page=2]{Figures/bif_bound.pdf}
  \hspace{0.05\linewidth}
  \includegraphics[width=0.4\linewidth,page=1]{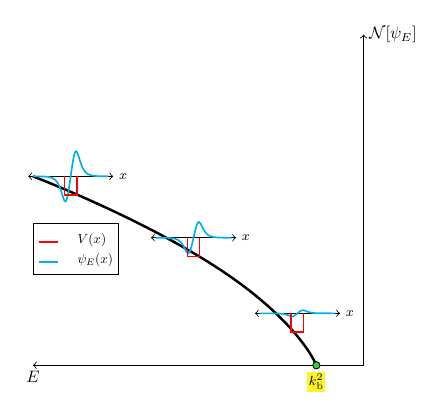}

  \caption{
  Schematic of bifurcation due to a bound state pole of $H_V$ at $k=i\kappa_b$, $\kappa_b>0$. 
  Left panel: Highlighted is $k_b=i\kappa_b$ at which $w(k_b)=0$. The inset shows the corresponding anti-symmetric eigenstate $U_\star$ of $H_V$. 
  Right panel: Curve of nonlinear bound states bifurcating from $\mathcal{N}=0$ at energy $E=E_b=k_b^2=-\kappa_b^2<0$.
  }
  \label{fig:boundstate}
\end{figure}

\subsection{Scattering resonance pole bifurcations}\label{sec:srpbif}
Figure~\ref{fig:resonance} describes the bifurcation arising from a scattering resonance pole at $H_V$ at $k_\star=i\kappa_r$, with $\kappa_r<0$. The curve of nonlinear bound states, $E\mapsto \psi_E$, emerges from  $(E,\mathcal{N})=(E_r,\mathcal{N}_\star)$,   where 
$ E_r=-\kappa_r^2\quad {\rm and}\quad \mathcal{N_\star}= \lim_{E\to E_r}\mathcal N[\psi_E] = 2\ \mathcal N[S_{k_{\rm r}^2}] = 2\times 4 |k_{\rm r}| = 8\ |k_{\rm r}|$, which is strictly positive in contrast to the zero limiting $L^2$ norm for the bound state pole scenario.
For $E$ less than and near $-\kappa_r^2$,  $\psi_E(x)$ is approximately equal to an anti-symmetric superposition of distant 1-soliton  ($\sech$)  profiles. 
\begin{figure}[!ht]
  \centering
  % --- column titles ---
  \makebox[0.4\linewidth][c]{\textbf{Linear}}%
  \hspace{0.05\linewidth}%
  \makebox[0.4\linewidth][c]{\textbf{Nonlinear}}\\[2pt]

  % --- swapped images ---
  \includegraphics[width=0.4\linewidth,page=2]{Figures/bif_rez.pdf}
  \hspace{0.05\linewidth}
  \includegraphics[width=0.4\linewidth,page=1]{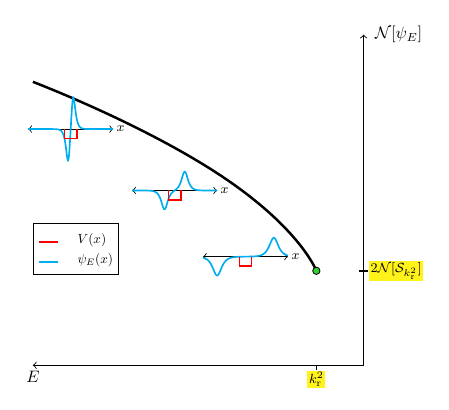}

  \caption{
  Schematic showing the 
  bifurcation of a nonlinear bound state $(\psi_\star,k_{\rm r}^2)$ from a scattering resonance (anti-bound state) pole of $H_V$. 
  Left panel: Highlighted is a resonance (anti-bound state) pole $k_r=i\kappa_r$, with $\kappa_r<0$, at which $w(k_r)=0$. The inset shows the corresponding anti-symmetric scattering resonance mode $U_\star$ of $H_V$. 
  Right panel: Curve of nonlinear bound states bifurcating at a strictly positive $L^2$ threshold $\mathcal{N}=2\mathcal{N}[\mathcal S_{{k_r}^2}]>0$ at energy $E=k_r^2=-\kappa_r^2<0$.
  }
  \label{fig:resonance}
\end{figure}

\subsection{Transmission resonance bifurcations}\label{sec:trbif}

Figure~\ref{fig:transmission} describes the scenario of bifurcation from a right transmission resonance $(U_\star,k_{\rm t}^2)$. With $k_{\rm t} = i \kappa_{\rm t}$, $\kappa_{\rm t} < 0$.  Near the bifurcation point, the solution is given by a single $\sech$ profile, very distant from $x=0$. Hence $\mathcal N[\psi_E] \to  \mathcal N[S_{k_{\rm t}^2}] = 4 |k_{\rm t}|>0$. Note that the character of the threshold $L^2$ norm in this case is half that for the case of a scattering resonance pole bifurcation. 
\begin{figure}[!ht]
  \centering
  % --- column titles ---
  \makebox[0.4\linewidth][c]{\textbf{Linear}}%
  \hspace{0.05\linewidth}%
  \makebox[0.4\linewidth][c]{\textbf{Nonlinear}}\\[2pt]

  % --- swapped images ---
  \includegraphics[width=0.4\linewidth,page=2]{Figures/bif_trans.pdf}
  \hspace{0.05\linewidth}
  \includegraphics[width=0.4\linewidth,page=1]{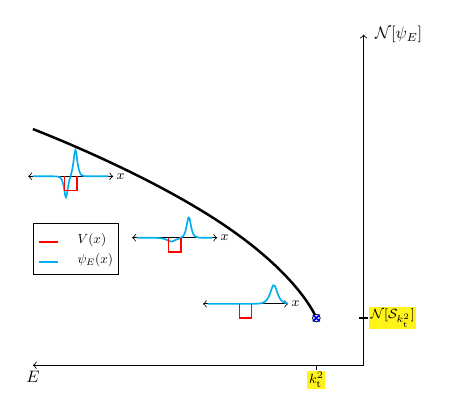}

  \caption{
  Schematic showing the 
  bifurcation of nonlinear bound states from a transmission resonance  $(U_\star,k_{\rm t}^2)$. The asymmetry of the transmission resonance mode is manifested in the nonlinear bound states.
  }
  \label{fig:transmission}
\end{figure}

 Analogously, there is a bifurcation from a left-going transmission resonances,  $k_{2} = i \kappa_{2}$, $\kappa_{2} > 0$, which traces out the same $\mathcal{N}$ vs $E$ curve.

\begin{rem}[Observation of topology changes in bifurcation diagrams; Figure~\ref{fig:coalesce}.]\label{rem:toch}
    Suppose that as $V$ is deformed, two scattering resonance poles (anti-bound states) on the imaginary axis in $\mathbb C_-$, $k_1=i\kappa_1$  and $k_2=i\kappa_2$,  approach each other, coalesce ($k_1=k_2)$ and then depart from  the imaginary axis ($k_1$ and $k_2=-\bar{k_1}$).
    In this case,  two disconnected bifurcation branches of nonlinear bound states at first approach
     a common point in the $E<0$ and $\mathcal{N}>0$ quadrant. Beyond the coalescence point of $k_1$ and $k_2$, there remains a single smooth connected bifurcation branch. We illustrate this in Figure~\ref{fig:coalesce}.

\end{rem}

\begin{figure}[!ht]
  \centering
  % --- small titles row ---
  \makebox[0.4\linewidth][c]{\textbf{Linear}}%
  \hspace{0.05\linewidth}%
  \makebox[0.4\linewidth][c]{\textbf{Nonlinear}}\\[2pt]

  % --- first row of images ---
  \includegraphics[width=0.4\linewidth,page=1]{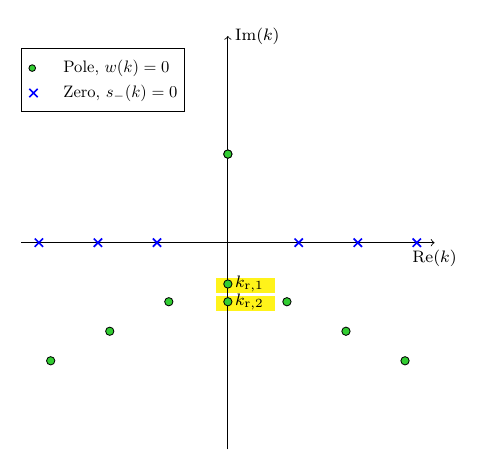}
  \hspace{0.05\linewidth}
  \includegraphics[width=0.4\linewidth,page=3]{Figures/bif_split.pdf}\\[6pt]

  % --- second row of images ---
  \includegraphics[width=0.4\linewidth,page=2]{Figures/bif_split.pdf}
  \hspace{0.05\linewidth}
  \includegraphics[width=0.4\linewidth,page=4]{Figures/bif_split.pdf}

  \caption{
  Topological change in the bifurcation diagram due to coalescence of scattering resonance poles. 
  Branches before collision drift off to $\pm\infty$ (left), while after collision they reorganize (right).
  See Remark~\ref{rem:toch}.
  }
  \label{fig:coalesce}
\end{figure}

\subsection{Bifurcations from threshold (zero-energy) resonances}
\label{sec:zen-res}
In Theorem \ref{thm:bifurcations-zero}, for symmetric potentials $V(x)=V(-x)$ with a threshold resonance ($w(0)=0$), either one even or one odd branch bifurcates from the threshold.
Any additional threshold-induced branches must be asymmetric and therefore occur in reflected pairs $\psi_2(x)=\psi_1(-x)$.
Consequently, symmetric potentials yield an odd number of bifurcating branches.
Figures~\ref{fig:zerobif2}–\ref{fig:zerobif3} illustrate cases for symmetric wells, while Figure~\ref{fig:zerobif1} illustrates a case for an asymmetric well. These examples correspond to potentials of different depth:

\begin{enumerate}
    \item \textbf{Left}, $\alpha< \alpha_\star$, $H_V$ has a scattering resonance, $w(i\kappa_-) = 0$, with $\kappa_-<0$,

\item  \textbf{Center}, $\alpha = \alpha_\star$, $H_V$ has a (zero energy) threshold resonance $w(0) = s_-(0) =0$, and

\item \textbf{Right}, $\alpha > \alpha_\star$, 
 $H_V$ supports a new bound state, $w(i\kappa_+) = 0$, with $\kappa_+>0$.
\end{enumerate}

We consider the 2-parameter family of continuous, compactly supported, single-well potential:
\begin{equation}\label{eq:W}
    V(x;\alpha,\beta)
    = -\alpha\, \frac{W(x;\beta)}{\max_x W(x;\beta)}, \qquad
    W(x;\beta) = \exp\!\Big(-\tfrac{1}{4}(x-\beta)^{2}\Big)
    \!\left[\tfrac{1}{2}\bigl(1+\cos(\pi x)\bigr)\right],
    \quad \lvert x \rvert \le 1,
\end{equation}
and $V(x;\alpha,\beta)=0$ for $\lvert x \rvert >1$. 
The parameter $\alpha>0$ controls the well depth, while $\beta$ shifts the center of the Gaussian envelope, introducing asymmetry.  
For each fixed $\beta$, we pick the value $\alpha_\star(\beta)$ at which the linear operator $H_V=-\partial_x^2+V$ develops a zero-energy (threshold) resonance.

\textbf{Figure~\ref{fig:zerobif2} --- Symmetric $V(\; \cdot \;;\alpha,\beta) \in L^1_{\rm comp} \cap C^0$ single well:}\\
First, we set $\beta=0$, making $V$ even.  
The operator $H_V$ acts invariantly on even and odd subspaces.  
In this case, the transmission resonances occur in pairs $k_1 = \bar k_2$, and when $k_1, k_2 \in i \mathbb R$ they induce nonlinear bound state branches that satisfy the mirror relation $\psi_1(x)=-\psi_2(-x)$.  At threshold, a single symmetric branch bifurcates from $(E,\mathcal N)=(0,0)$.

\begin{figure}[!ht]
\centering

% === Row 1: less / eq / greater ===
\begin{minipage}[t]{0.29\linewidth}\centering
  \includegraphics[width=\linewidth]{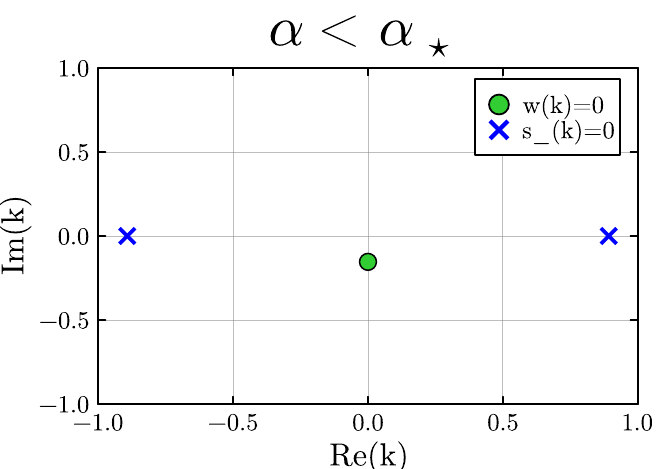}
\end{minipage}
\hspace{0.05\linewidth}
\begin{minipage}[t]{0.29\linewidth}\centering
  \includegraphics[width=\linewidth]{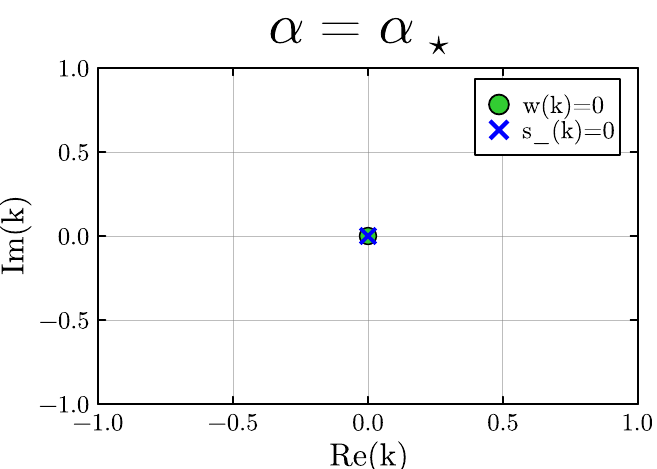}
\end{minipage}
\hspace{0.05\linewidth}
\begin{minipage}[t]{0.29\linewidth}\centering
  \includegraphics[width=\linewidth]{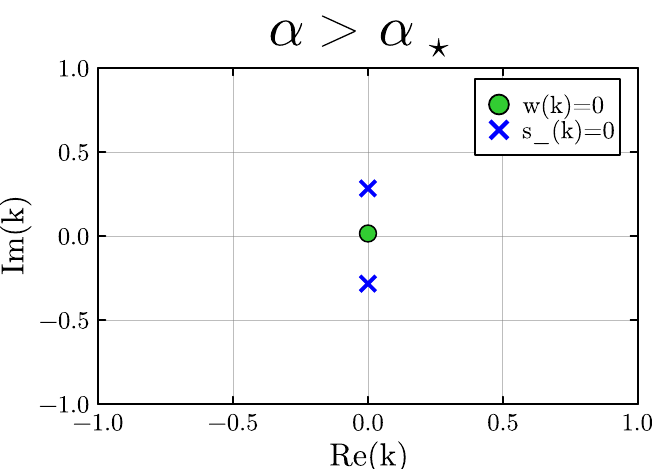}
\end{minipage}
\\[-0.0em]

% === Row 2: L2_sub / L2_crit / L2_super ===
\begin{minipage}[t]{0.29\linewidth}\centering
  \includegraphics[width=\linewidth]{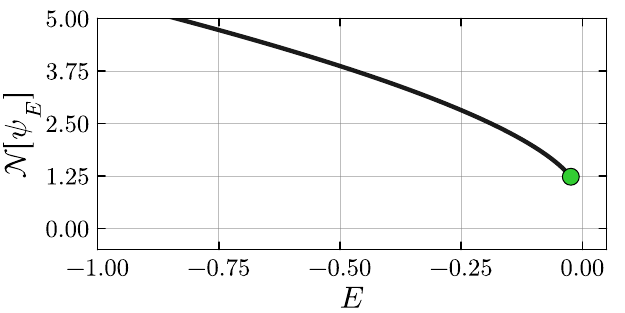}
\end{minipage}
\hspace{0.05\linewidth}
\begin{minipage}[t]{0.29\linewidth}\centering
  \includegraphics[width=\linewidth]{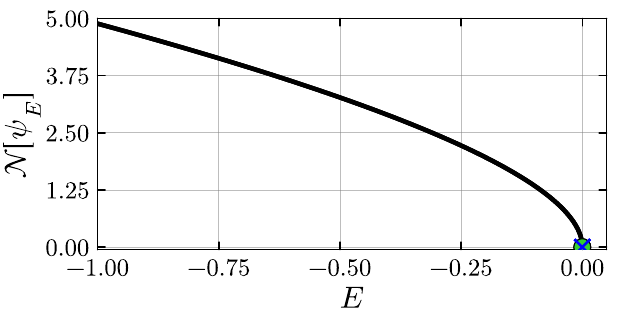}
\end{minipage}
\hspace{0.05\linewidth}
\begin{minipage}[t]{0.29\linewidth}\centering
  \includegraphics[width=\linewidth]{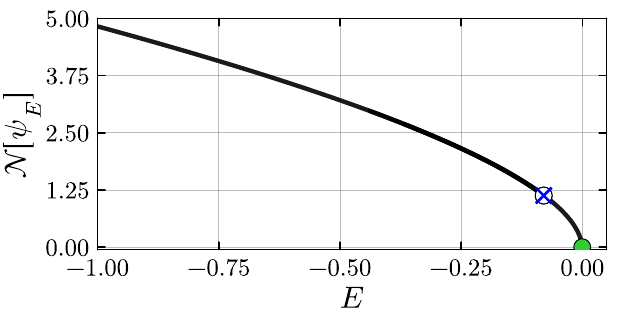}
\end{minipage}
\\[-0.2em]

% === Row 3: H1_sub / H1_crit / H1_super ===
\begin{minipage}[t]{0.29\linewidth}\centering
  \includegraphics[width=\linewidth]{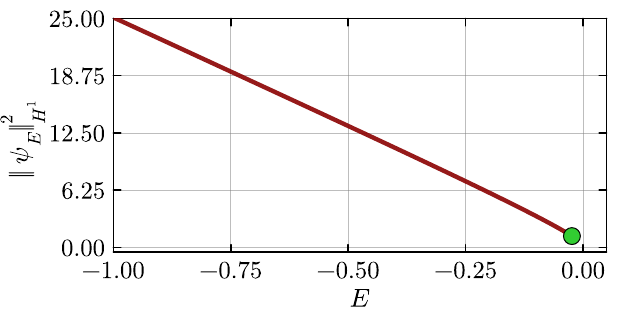}
\end{minipage}
\hspace{0.05\linewidth}
\begin{minipage}[t]{0.29\linewidth}\centering
  \includegraphics[width=\linewidth]{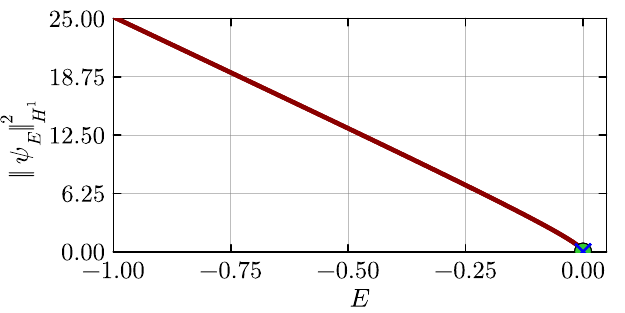}
\end{minipage}
\hspace{0.05\linewidth}
\begin{minipage}[t]{0.29\linewidth}\centering
  \includegraphics[width=\linewidth]{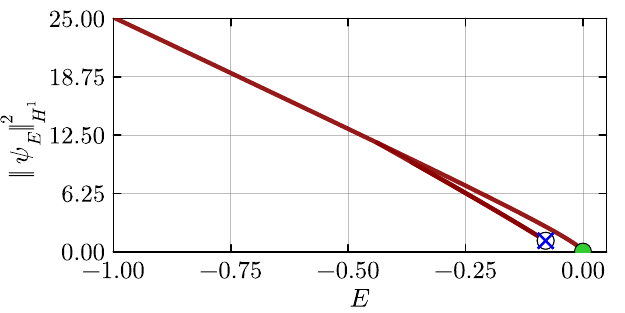}
\end{minipage}
\\[-0.2em]

% === Row 4: plots_sub / plots_crit / plots_super ===
\begin{minipage}[c]{0.29\linewidth}\centering
  \hspace*{0.4em}
  \includegraphics[width=\linewidth]{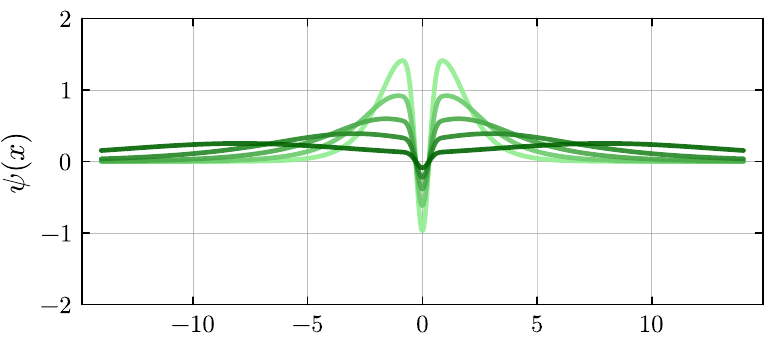}
\end{minipage}
\hspace{0.05\linewidth}
\begin{minipage}[c]{0.29\linewidth}\centering
  \hspace*{0.4em}
  \includegraphics[width=\linewidth]{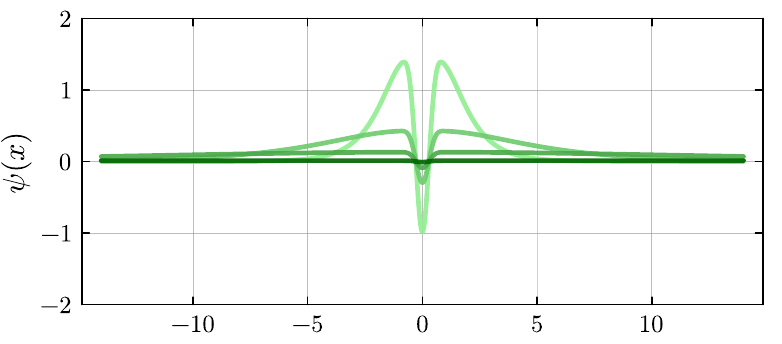}
\end{minipage}
\hspace{0.05\linewidth}
\begin{minipage}[c]{0.29\linewidth}\centering
  \includegraphics[width=\linewidth]{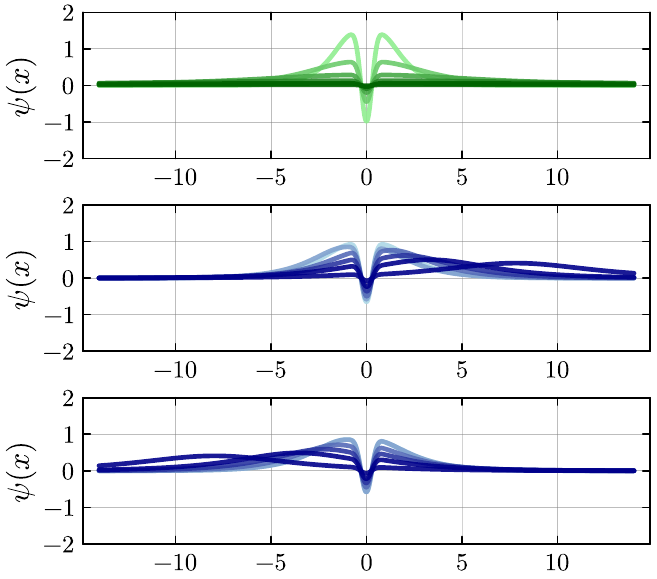}
\end{minipage}
\caption{
Symmetric potential $V(x;\alpha,\beta=0)$ from~\eqref{eq:W}.  
Top row: zeros of $w(k)$~\protect\greenCirc\ and of $s_-(k)$~\protect\blueXn.  
Rows~2–3: $\mathcal N$- and $H^1$-norm bifurcation diagrams.  
Bottom row: corresponding nonlinear profiles. The $H^1$-norm diagram is shown to distinguish branches that overlap in the $L^2$ picture.  
At $\alpha=\alpha_\star$, there is a single branch.
For $\alpha>\alpha_\star$, two transmission-induced asymmetric branches related by $\psi_1(x)=-\psi_2(-x)$ (of the same norms) merge into the symmetric branch near $E \approx -0.45$.  
Solution profiles $\psi(x;E)$ are plotted darker for $E$ closer to the bifurcation points~\protect\blueX\ and \protect\greenCirc.
}
% (\emph{lighter near $E=-1.0$, except in the $\alpha>\alpha_\star$ case where the blue profiles meet the symmetric branch at a bifurcation point $E_{\rm min} \approx -0.45$}).
\label{fig:zerobif2}
\end{figure}

\textbf{Figure~\ref{fig:zerobif3} --- Symmetric square well potential:}\\
Next, we consider the square well
\[
V(x;\alpha) = 
\begin{cases}
-\alpha, & \lvert x \rvert \le1,\\[4pt]
0, & \lvert x \rvert >1.
\end{cases}
\]
 When there is a threshold resonance ($\alpha=\alpha_\star$), three branches emerge from $(E,\mathcal N) = (0,0)$: one odd-symmetric and two reflected asymmetric branches $\psi_1(x)$ and $\psi_2(x) = -\psi_1(-x)$.

\begin{figure}[!ht]
\centering

% === First row: potentials ===
\begin{minipage}[t]{0.28\linewidth}\centering
  \includegraphics[width=\linewidth]{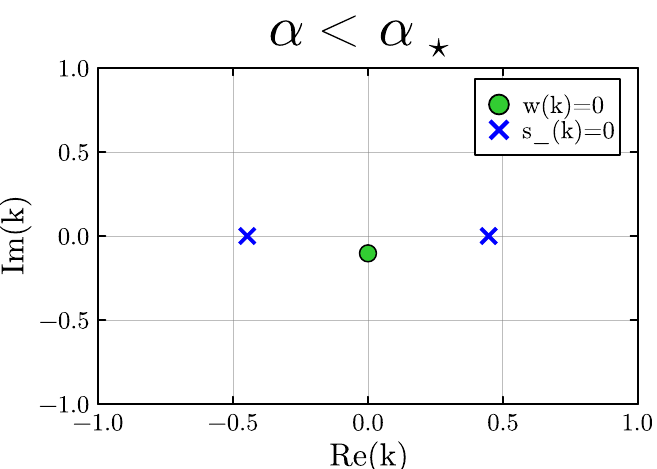}
\end{minipage}
\hspace{0.05\linewidth}
\begin{minipage}[t]{0.28\linewidth}\centering
  \includegraphics[width=\linewidth]{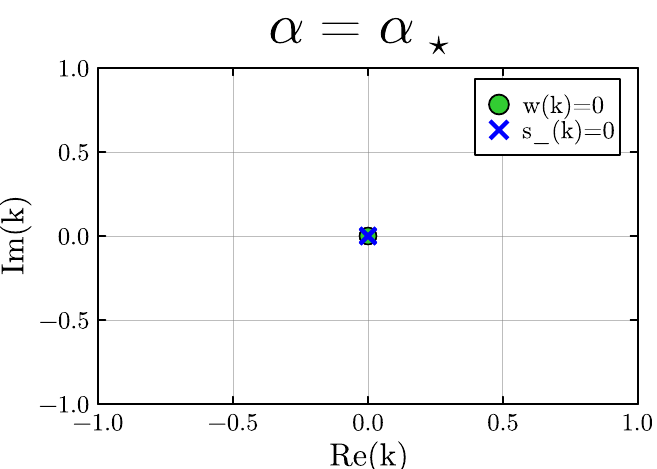}
\end{minipage}
\hspace{0.05\linewidth}
\begin{minipage}[t]{0.28\linewidth}\centering
  \includegraphics[width=\linewidth]{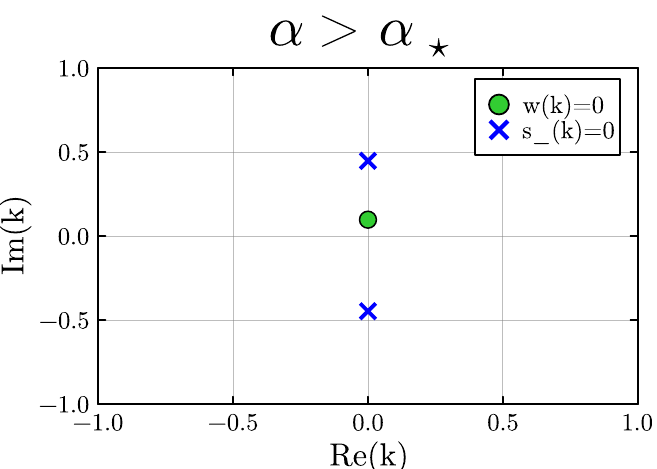}
\end{minipage}
\\[-0.0em]

% === Second row: bifurcation diagrams ===
\begin{minipage}[t]{0.28\linewidth}\centering
  \includegraphics[width=\linewidth]{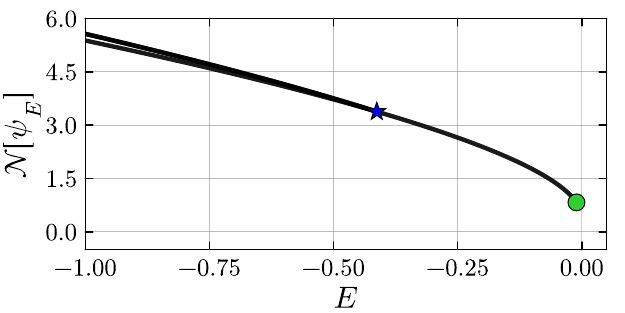}
\end{minipage}
\hspace{0.05\linewidth}
\begin{minipage}[t]{0.28\linewidth}\centering
  \includegraphics[width=\linewidth]{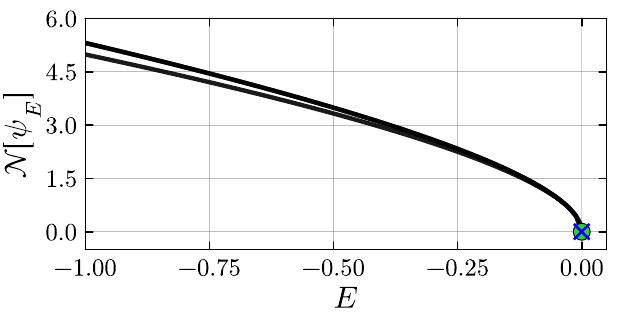}
\end{minipage}
\hspace{0.05\linewidth}
\begin{minipage}[t]{0.28\linewidth}\centering
  \includegraphics[width=\linewidth]{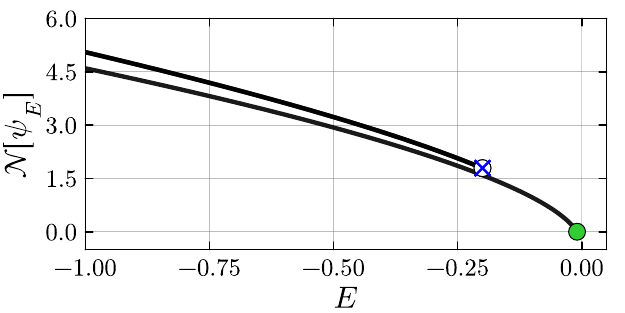}
\end{minipage}
\\[-0.2em]

% === Third row: wave profiles ===
\begin{minipage}[t]{0.28\linewidth}\centering
  \includegraphics[width=\linewidth]{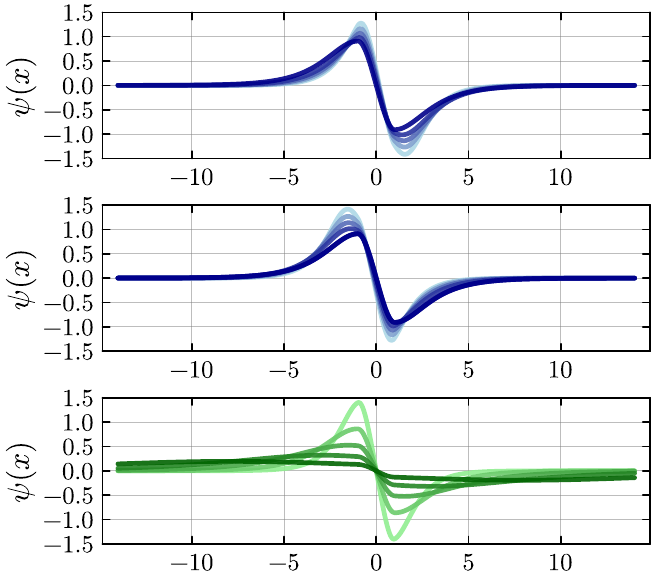}
\end{minipage}
\hspace{0.05\linewidth}
\begin{minipage}[t]{0.28\linewidth}\centering
  \includegraphics[width=\linewidth]{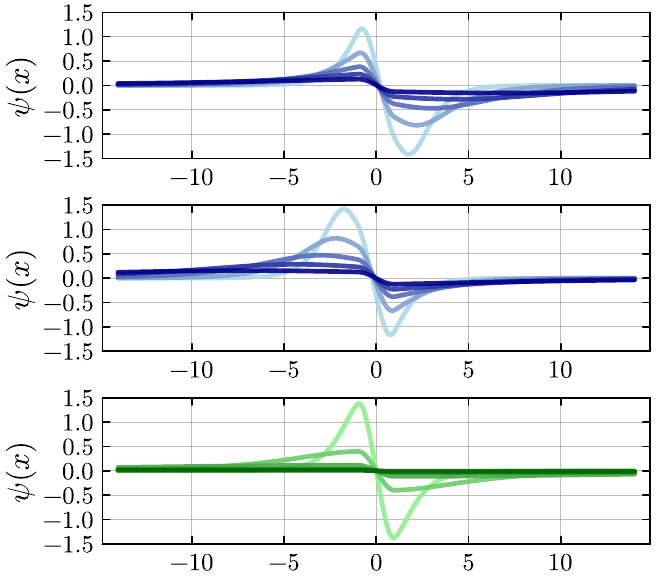}
\end{minipage}
\hspace{0.05\linewidth}
\begin{minipage}[t]{0.28\linewidth}\centering
  \includegraphics[width=\linewidth]{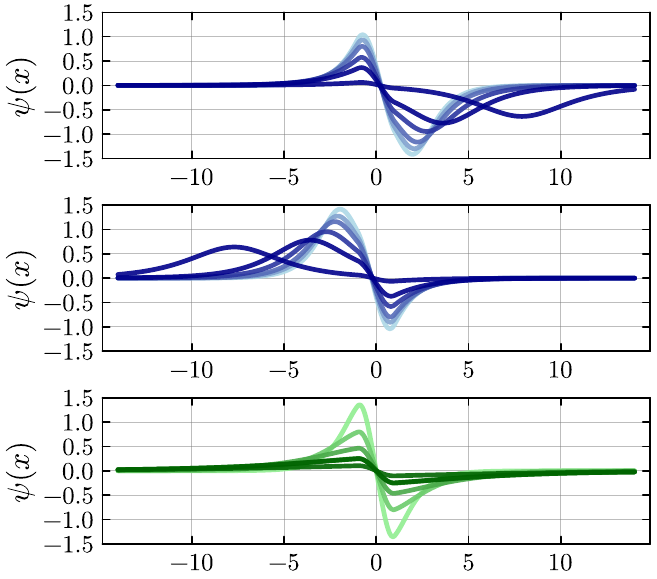}
\end{minipage}

\caption{
Symmetric square well potential $V(x)=-\alpha\chi_{[-1,1]}(x)$.  
Top row: zeros of $w(k)$~\protect\greenCirc\ and of $s_-(k)$~\protect\blueXn.  
Row~2: $\mathcal N$ bifurcation diagram.  
Rows~3–4: nonlinear profiles and associated bifurcation curves.  
At $\alpha=\alpha_\star$, three branches bifurcate from the zero-energy resonance: one odd-symmetric and two reflected asymmetric branches $\psi_1(x)$ and $\psi_2(-x)$.  
These persist for $\alpha<\alpha_\star$ as bifurcations at the discontinuity points of $V$ where $\psi_E'(\pm b)=0$, and for $\alpha>\alpha_\star$ as bifurcations from transmission resonances on the imaginary axis.  
Solution profiles $\psi(x;E)$ are plotted darker for $E$ closer to the bifurcation points~$\star$, \protect\blueX, and \protect\greenCirc\ (\emph{lighter near $E = -1.0$}).
}
\label{fig:zerobif3}
\end{figure}

\textbf{Figure~\ref{fig:zerobif1} --- Asymmetric $V(\; \cdot \; ;\alpha,\beta) \in L^1_{\rm comp} \cap C^0$ single well:}\\
Now, we fix $\beta=-11$ for $V$ in \eqref{eq:W} to make it asymmetric. As $\alpha$ increases through $\alpha_\star$, a scattering resonance and transmission resonance on $i \mathbb R$ collide at $k=0$ to form a zero-energy threshold resonance.  For $\alpha<\alpha_\star$, the upper branch originates from the scattering pole \greenCirc\ and the lower from the transmission zero \blueX;
for $\alpha>\alpha_\star$, these associations are interchanged.

\begin{figure}[!ht]
\centering

% === First row ===
\begin{minipage}[t]{0.29\linewidth}\centering
  \includegraphics[width=\linewidth]{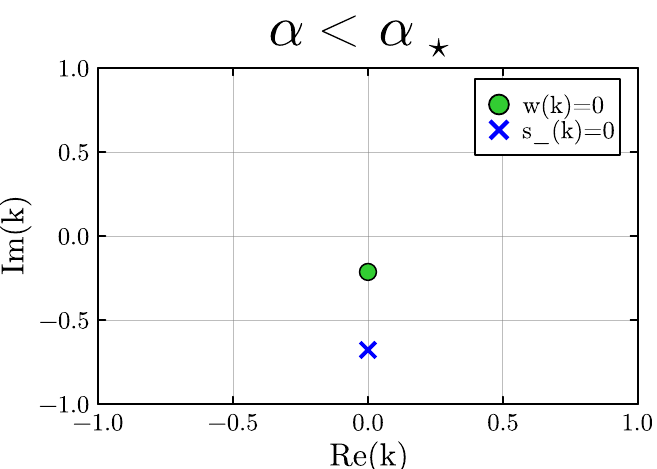}
\end{minipage}
\hspace{0.05\linewidth}
\begin{minipage}[t]{0.29\linewidth}\centering
  \includegraphics[width=\linewidth]{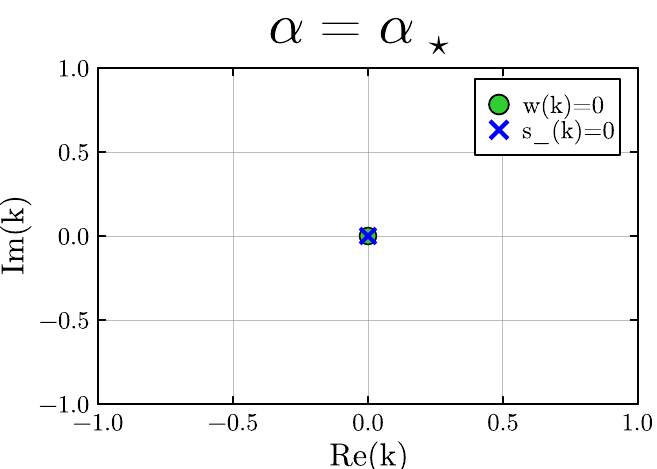}
\end{minipage}
\hspace{0.05\linewidth}
\begin{minipage}[t]{0.29\linewidth}\centering
  \includegraphics[width=\linewidth]{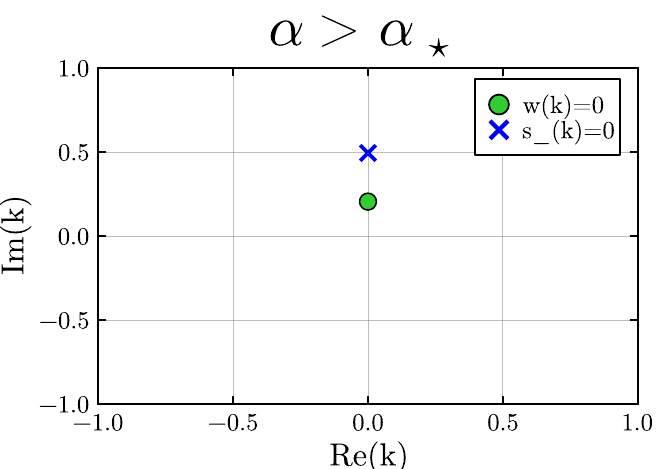}
\end{minipage}
\\[-0.0em]

% === Second row ===
\begin{minipage}[t]{0.29\linewidth}\centering
  \includegraphics[width=\linewidth]{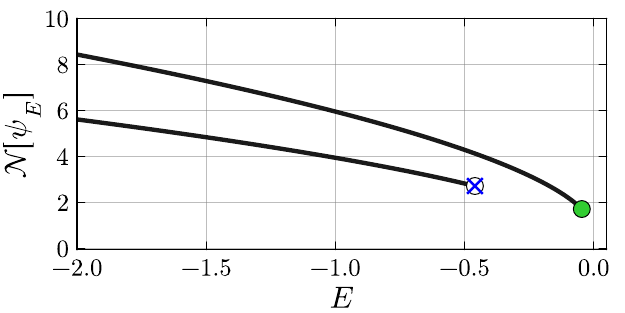}
\end{minipage}
\hspace{0.05\linewidth}
\begin{minipage}[t]{0.29\linewidth}\centering
  \includegraphics[width=\linewidth]{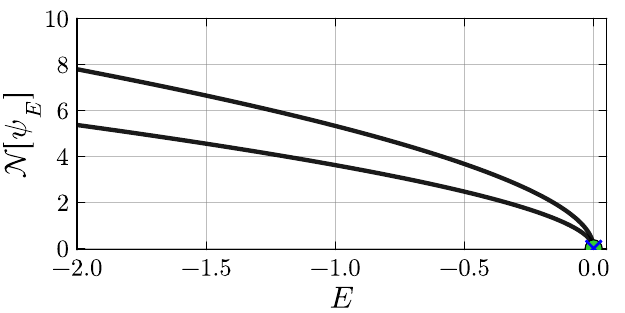}
\end{minipage}
\hspace{0.05\linewidth}
\begin{minipage}[t]{0.29\linewidth}\centering
  \includegraphics[width=\linewidth]{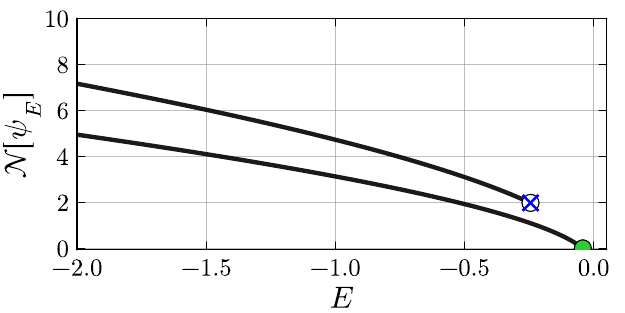}
\end{minipage}
\\[-0.15em]

% === Third row ===
\begin{minipage}[t]{0.29\linewidth}\centering
  \includegraphics[width=\linewidth]{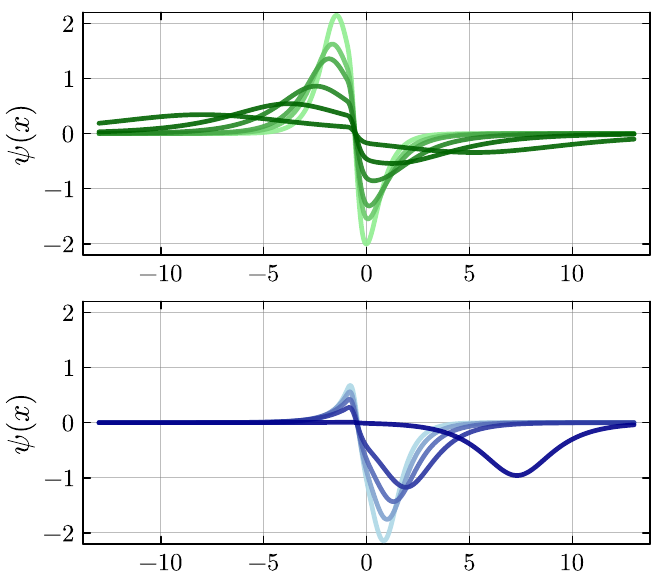}
\end{minipage}
\hspace{0.05\linewidth}
\begin{minipage}[t]{0.29\linewidth}\centering
  \includegraphics[width=\linewidth]{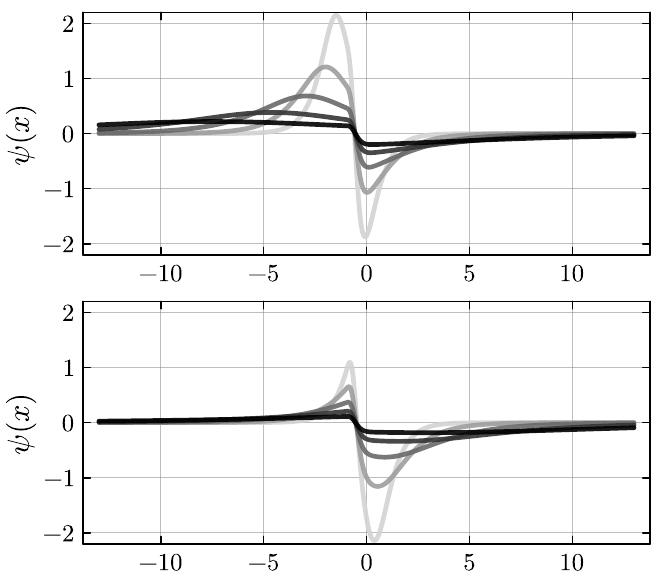}
\end{minipage}
\hspace{0.05\linewidth}
\begin{minipage}[t]{0.29\linewidth}\centering
  \includegraphics[width=\linewidth]{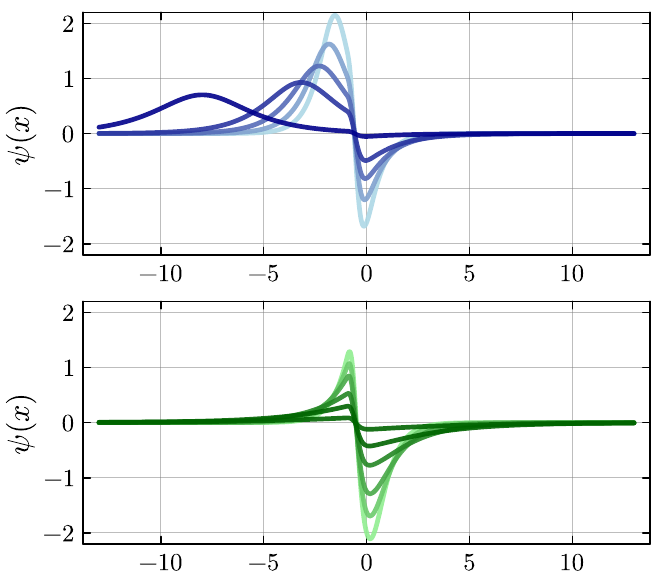}
\end{minipage}

\caption{
Asymmetric potential $V(x;\alpha,\beta=-11)$ from~\eqref{eq:W}.  
Row~1: scattering data—zeros of $w(k)$~\protect\greenCirc\ and of $s_-(k)$~\protect\blueXn.  
Row~2: nonlinear bifurcation diagram.  
Rows~3–4: nonlinear bound-state profiles.  
$H_V$ admits a threshold resonance at $\alpha=\alpha_\star=24.04031$.  
As $\alpha$ passes through $\alpha_\star$, the scattering and transmission resonances (\protect\greenCirc, \protect\blueX) exchange roles.  
Solution profiles $\psi(x;E)$ are plotted darker for $E$ closer to the bifurcation points~\protect\blueX\ and \protect\greenCirc\ (\emph{lighter near $E = -2.0$}).}

    \label{fig:zerobif1}
\end{figure}

\section{Discussion and future directions}\label{sec:future}
In this article we have demonstrated that nonlinear bound states of NLS/GP arise via bifurcation from scattering resonances and transmission resonances of the Schr\"odinger operator $H_V=-\partial_x^2+V$, a generalization of the bifurcation from bound states of $H_V$. Our study elucidates the physical heuristic that the self-consistent potential of NLS/GP, $V(x)- |\psi(x)|^2$ can, at some threshold with respect to the size of $\psi$, convert certain resonance states of $H_V$ that are merely $L^2(K)$, for $K$ compact,  into  nonlinear bound state solutions in $L^2(\mathbb R)$. We next mention some future directions of interest.

\begin{enumerate}[label=(\arabic*), leftmargin=*, labelwidth=0pt, labelsep=0.4em]
 \item \textbf{Dynamical Stability}:   In \cite{rose1988bound}, it was proved that the branch which bifurcates from a linear ground state is nonlinearly orbitally stable; see also \cite{weinstein1985modulational,weinstein1986lyapunov,grillakis1987stability,jackson2004geometric,kirr2008symmetry,weinstein-survey:15}. 
 Asymptotic stability and scattering were addressed in, for example,  \cite{sw1990,sw1992,pillet-wayne1997,sw2004}.  It is natural to consider the stability and instability properties of branches of states which arise from linear scattering and transmission resonances.

 \item \textbf{Non-compactly supported potentials:} Our gluing construction depends on the potential $V$ having compact support. Note that if $V$ decays exponentially then the scattering data functions  $w(k)$
 and $s_\pm(k)$ are analytic in a strip. We may then consider the bound state and scattering poles within this strip of analyticity.   Do analogous bifurcation results hold for sufficiently decaying, non-compactly supported, potentials? 
Is there a sharp decay rate on $V(x)$ for such results to hold? Variational results in \cite{kirr2018global} suggest that the drifting of solitons to infinity, which is a hallmark of our resonance-type bifurcations, cannot occur for potentials which decay sufficiently slowly.

\item {\bf Periodic or discrete NLS/GP}: Are there analogous phenomena in the NLS/GP, where the underlying linear potential is periodic plus a compactly supported defect potential \cite{dvw2015}?
And, the analogous question for equations of  discrete nonlinear Schr\"odinger type. 

 \item \textbf{Multiple bifurcations from a threshold resonance}: Give a full explanation of the possibility of multiple bifurcations of states from a threshold (zero energy) resonances. 
 Investigate the possible applicability of the bifurcation strategy used in the study of spectral band edge bifurcations
\cite{ilan-weinstein:10} for NLS with a periodic potential $V$ and for discrete NLS-type equations \cite{jenkinson:16,jenkinson:17}.

\item \textbf{Variational perspective.}  
Consider the ground states in $H^1_0(\mathbb R_+)$, defined as
\[
    \inf_{\substack{\mathcal N[u]=N, \\ u\in H^1_0(\mathbb R_+)}} \mathcal H^V[u]
\]
with Lagrange multiplier $E=E(N)$. Or, equivalently, the ground states in $H^1_{\rm odd}(\mathbb R)$ for symmetric $V(x) = V(-x)$ potentials. We expect families of constrained minimizers, similar to those displayed in Figures~\ref{fig:boundstate} and \ref{fig:resonance}. When $H_V$ has no bound state, the `ground state' curves terminate at a positive $\mathcal N$-threshold.  
We conjecture the existence of a critical $\mathcal N_{\mathrm{cr}}$ below which minimizing sequences lose precompactness via the ``drifting-to-infinity’’ mechanism~\cite{lions1984concentration}.  
Along our resonance-induced bifurcation curves, the nonlinear bound states similarly move farther from the support of $V$ as the bifurcation point is approached.  
Numerically (e.g., Figure~\ref{fig:coalesce}), the $L^2$ excitation threshold does not always coincide with a resonance bifurcation. The minimizing state could also plausibly jump from branch to branch as $\mathcal N$ varies.
\item 
\textbf{Time-dependent scattering dynamics}:  Do our new states participate, even as long time transients, in time-dependent scattering of NLS/GP. Are such states visible in 
 delta potential barrier potential scattering  \cite{holmer2007fast} and scattering for trapping potentials \cite{goodman2004strong}?
\item \textbf{Higher dimensions:}  Let  $V$ denote a sufficiently decaying potential well on $\mathbb R^3$. Then, for  $\alpha_\star>0$ sufficiently small $H_V = -\Delta + \alpha V$ has no point spectrum for  $\alpha < \alpha_\star$, and has point spectrum for  $\alpha > \alpha_\star$ \cite{lieb2001analysis}. 
In the latter regime,  nonlinear bound states bifurcate from the eigenstates of $H_V$, and the family arising from the ground state is dynamically orbitally stable \cite{rose1988bound} and asymptotically stable \cite{sw1990,sw1992,sw2004,weinstein-survey:15} for the NLS/GP dynamics. As suggested in \cite{rose1988bound}, does the focusing nonlinearity of NLS/GP induce stable nonlinear bound states from scattering resonances modes of $H_V$ for $0<\alpha<\alpha_\star$?
\end{enumerate}

\vfill
\printbibliography

@book{sulemsulem1999,
  title={The nonlinear Schr{\"o}dinger equation: Self-Focusing and Wave Collapse},
  author={Sulem, C. and Sulem, P. L.},
  series={Applied Mathematical Sciences},
  publisher={Springer},
  year={1999}
}

@book{fibich2015nonlinear,
  title={The nonlinear Schr{\"o}dinger equation},
  author={Fibich, Gadi},
  year={2015},
  publisher={Springer}
}

@book{cazenave2003semilinear,
  title={Semilinear Schr\"odinger equations},
  author={Cazenave, Thierry},
  year={2003},
  publisher={American Mathematical Society}
}

@article{oh1989cauchy,
  title={Cauchy problem and Ehrenfest's law of nonlinear Schr{\"o}dinger equations with potentials},
  author={Oh, Yong-Geun},
  journal={Journal of Differential Equations},
  volume={81},
  number={2},
  pages={255--274},
  year={1989},
  publisher={Elsevier}
}

@article{esy2010,
  title={Derivation of the Gross-Pitaevskii equation for the dynamics of Bose-Einstein condensate},
  author={Erd{\H{o}}s, L{\'a}szl{\'o} and Schlein, Benjamin and Yau, Horng-Tzer},
  journal={Annals of mathematics},
  pages={291--370},
  year={2010},
  publisher={JSTOR}
}

@inproceedings{lions1984concentration,
  title={The concentration-compactness principle in the Calculus of Variations. The locally compact case, part 1.},
  author={Lions, Pierre-Louis},
  booktitle={Annales de l'Institut Henri Poincar{\'e} C, Analyse non lin{\'e}aire},
  volume={1},
  number={2},
  pages={109--145},
  year={1984},
  organization={Elsevier}
}

@article{simon1976bound,
  title={The bound state of weakly coupled Schr\"odinger operators in one and two dimensions},
  author={Simon, Barry},
  journal={Annals of Physics},
  volume={97},
  pages={279--288},
  year={1976}
}

@article{pelinovsky2004convergence,
  title={Convergence of Petviashvili's iteration method for numerical approximation of stationary solutions of nonlinear wave equations},
  author={Pelinovsky, Dmitry E and Stepanyants, Yury A},
  journal={SIAM Journal on Numerical Analysis},
  volume={42},
  number={3},
  pages={1110--1127},
  year={2004},
  publisher={SIAM}
}

@article{petviashvili1976equation,
  title={Equation of an extraordinary soliton},
  author={Petviashvili, Vladimir I},
  journal={Fizika plazmy},
  volume={2},
  pages={469--472},
  year={1976}
}

@article{rose1988bound,
  title={On the bound states of the nonlinear Schr{\"o}dinger equation with a linear potential},
  author={Rose, Harvey A and Weinstein, Michael I},
  journal={Physica D: Nonlinear Phenomena},
  volume={30},
  number={1-2},
  pages={207--218},
  year={1988},
  publisher={Elsevier}
}

@book{reed1979iii,
  title={III: Scattering Theory},
  author={Reed, Michael and Simon, Barry},
  volume={3},
  year={1979},
  publisher={Elsevier}
}

@article{deift1979inverse,
  title={Inverse scattering on the line},
  author={Deift, Percy and Trubowitz, Eugene},
  journal={Commun. Pure Appl. Math.},
  volume={32},
  number={2},
  pages={121--251},
  year={1979},
  publisher={Wiley}
}

@article{korotyaev2005inverse,
  title={Inverse resonance scattering on the real line},
  author={Korotyaev, Evgeny},
  journal={Inverse Problems},
  volume={21},
  number={1},
  pages={325--341},
  year={2005},
  publisher={IOP Publishing}
}

@book{dyatlov2019mathematical,
  title={Mathematical theory of scattering resonances},
  author={Dyatlov, Semyon and Zworski, Maciej},
  volume={200},
  year={2019},
  publisher={American Mathematical Soc.}
}

@article{BreitWigner1936,
  author  = {Gregory Breit and Eugene Wigner},
  title   = {Capture of Slow Neutrons},
  journal = {Physical Review},
  volume  = {49},
  pages   = {519--531},
  year    = {1936},
  doi     = {10.1103/PhysRev.49.519}
}

@article{Islam1966,
  author  = {M. M. Islam},
  title = {High-Energy $\pi p$ Scattering and the $\pi$–$\pi$, $J=0$, $T=0$ Antibound State},
  journal = {Physical Review},
  volume  = {147},
  number  = {4},
  pages   = {1144--1147},
  year    = {1966},
  doi     = {10.1103/PhysRev.147.1144}
}

@book{LandauLifshitzQM,
  author    = {L. D. Landau and E. M. Lifshitz},
  title     = {Quantum Mechanics: Non-Relativistic Theory},
  edition   = {3rd},
  publisher = {Pergamon Press},
  year      = {1977}
}

@book{Newton1982,
  author    = {Roger G. Newton},
  title     = {Scattering Theory of Waves and Particles},
  edition   = {2nd},
  publisher = {Springer-Verlag},
  address   = {New York},
  year      = {1982},
  doi       = {10.1007/978-3-642-88128-2}
}

@article{OhanianGinsburg1974,
  author  = {Hans C. Ohanian and Carl G. Ginsburg},
  title   = {Antibound `States' and Resonances},
  journal = {American Journal of Physics},
  volume  = {42},
  number  = {4},
  pages   = {310--315},
  year    = {1974},
  doi     = {10.1119/1.1987678}
}

@article{weinstein1985modulational,
  title={Modulational stability of ground states of nonlinear Schr{\"o}dinger equations},
  author={Weinstein, Michael I},
  journal={SIAM Journal on Mathematical Analysis},
  volume={16},
  number={3},
  pages={472--491},
  year={1985},
  publisher={SIAM}
}

@article{weinstein1986lyapunov,
  title={Lyapunov stability of ground states of nonlinear dispersive evolution equations},
  author={Weinstein, Michael I},
  journal={Communications on Pure and Applied Mathematics},
  volume={39},
  number={1},
  pages={51--67},
  year={1986},
  publisher={Wiley}
}

@article{grillakis1987stability,
  title={Stability theory of solitary waves in the presence of symmetry, I},
  author={Grillakis, Manoussos and Shatah, Jalal and Strauss, Walter},
  journal={Journal of Functional Analysis},
  volume={74},
  number={1},
  pages={160--197},
  year={1987},
  publisher={Elsevier}
}

@article{jackson2004geometric,
  title={Geometric Analysis of Bifurcation and Symmetry Breaking in a Gross—Pitaevskii Equation},
  author={Jackson, Russell K and Weinstein, Michael I},
  journal={Journal of Statistical Physics},
  volume={116},
  number={1},
  pages={881--905},
  year={2004},
  publisher={Springer}
}

@article{kirr2008symmetry,
  title={Symmetry-breaking bifurcation in nonlinear Schr{\"o}dinger/Gross--Pitaevskii equations},
  author={Kirr, EW and Kevrekidis, PG and Schlizerman, E and Weinstein, M I},
  journal={SIAM Journal on Mathematical Analysis},
  volume={40},
  number={2},
  pages={566--604},
  year={2008},
  publisher={SIAM}
}

@article{weinstein-survey:15,
  author = {Weinstein, M. I.},
  title = {Localized states and their dynamics in the nonlinear Schroedinger / Gross-Pitaeveskii equation},
  journal = {Frontiers in Applied Dynamics: Reviews and Tutorials},
  year = {2015}
}

@article{sw1990,
  title={Multichannel nonlinear scattering for nonintegrable equations},
  author={Soffer, A. and Weinstein, M.I.},
  journal={Commun. Math. Phys.},
  volume={133},
  pages={119–146},
  year={1990}
}

@article{sw1992,
  title={Multichannel nonlinear scattering and stability II. The case of anisotropic potential and data},
  author={Soffer, A. and Weinstein, M.I.},
  journal={J. Differential Equations},
  volume={98},
  pages={376--390},
  year={1992}
}

@article{pillet-wayne1997,
  title={Invariant manifolds for a class of dispersive, Hamiltonian, partial differential equations},
  author={Pillet, C.-A. and Wayne, C. E.},
  journal={J. Differential Equations},
  volume={141},
  pages={310--326},
  year={1997}
}

@article{sw2004,
  title={Selection of the ground state for nonlinear Schroedinger equations},
  author={Soffer, A. and Weinstein, M.I.},
  journal={Reviews in Mathematical Physics},
  volume={16},
  pages={977--1071},
  year={2004}
}

@article{kirr2018global,
  title={The global bifurcation picture for ground states in nonlinear Schrodinger equations},
  author={Kirr, Eduard and Natarajan, Vivek},
  journal={arXiv preprint arXiv:1811.05716},
  year={2018}
}

@article{dvw2015,
  title={Homogenized description of defect modes in periodic structures with localized defects},
  author={V. Duchêne and  I. Vukićević and M. I. Weinstein},
  journal={Commun. Math. Sci.},
  volume={13},
  number={3},
  pages={777--823},
  year={2015}
}

@article{ilan-weinstein:10,
  author = {Ilan, B. and Weinstein, M. I.},
  title = {Band-Edge Solitons, Nonlinear Schr\"odinger/Gross–Pitaevskii Equations, and Effective Media},
  journal = {Multiscale Modeling \& Simulation},
  volume = {8},
  number = {4},
  pages = {1055--1101},
  year = {2010},
  doi = {10.1137/090769417}
}

@article{jenkinson:16,
  author = {Jenkinson, M. and Weinstein, M. I.},
  title = {On-site and off-site solitary waves of the discrete nonlinear Schroedinger equation in multiple dimensions},
  journal = {Nonlinearity},
  volume = {29},
  number = {1},
  year = {2016}
}

@article{jenkinson:17,
  author = {Jenkinson, M. and Weinstein, M. I.},
  title = {Discrete Solitary Waves in Systems with Nonlocal Interactions and the Peierls-Nabarro Barrier},
  journal = {Communications in Mathematical Physics},
  volume = {29},
  number = {1},
  year = {2017}
}

@article{holmer2007fast,
  title={Fast soliton scattering by delta impurities},
  author={Holmer, Justin and Marzuola, Jeremy and Zworski, Maciej},
  journal={Comm. Math. Phys.},
  volume={274},
  number={1},
  pages={187--216},
  year={2007},
  publisher={Springer}
}

@article{goodman2004strong,
  title={Strong NLS soliton--defect interactions},
  author={Goodman, Roy H and Holmes, Philip J and Weinstein, Michael I},
  journal={Physica D: Nonlinear Phenomena},
  volume={192},
  number={3-4},
  pages={215--248},
  year={2004},
  publisher={Elsevier}
}

@book{lieb2001analysis,
  title={Analysis},
  author={Lieb, Elliott H and Loss, Michael},
  volume={14},
  year={2001},
  publisher={American Mathematical Soc.}
}
\end{document}